\documentclass[preprint, superscriptaddress,a4paper,longbibliography]{revtex4-1}
\usepackage[utf8]{inputenc}
\usepackage{fullpage}
\usepackage{tikz}
\usepackage{amsthm,amsmath,amssymb}
\usepackage{hyperref}
\usepackage[capitalize]{cleveref}

\newtheorem{lemma}{Lemma}

\newtheorem{theorem}{Theorem}
\newtheorem{corollary}{Corollary}
\theoremstyle{definition}

\newcommand{\tr}{\operatorname{Tr}}
\newcommand{\id}{\mathrm{Id}}

\newcommand{\ket}[1]{\vert #1 \rangle}

\usepackage{etoolbox}
\AtBeginEnvironment{tikzpicture}{\catcode `$=3 } 
\AtBeginEnvironment{align}{\catcode `$=3 } 
\AtBeginEnvironment{align*}{\catcode `$=3 } 

\usetikzlibrary{decorations.markings,calc}

\tikzset{
  tensor/.style={
    inner sep = 0.055cm,
    shape = circle,
    draw,
    fill
  },
  longtensor/.style={
    draw,
    fill,
    rounded corners=0.055cm,
    inner sep = 0.055cm,
    minimum width=1.31cm
  },
  verylongtensor/.style={
    draw,
    fill,
    rounded corners=0.055cm,
    inner sep = 0.055cm,
    minimum width=2.31cm
  },
  every picture/.style = {
    baseline={([yshift=-.5ex]current bounding box.center)}, 
    transform shape,
    font=\scriptsize
  }
}

\begin{document}

\title{Normal projected entangled pair states generating the same state}

\author{Andras Molnar}
\affiliation{Max-Planck-Institut für Quantenoptik, Hans-Kopfermann-Str. 1, D-85748 Garching, Germany}
\author{José Garre-Rubio}
\affiliation{Dpto. An\'alisis Matem\'atico y Matem\'atica Aplicada, Universidad Complutense de Madrid, 28040 Madrid, Spain}
\affiliation{Instituto de Ciencias Matem\'aticas, Campus Cantoblanco UAM, C/ Nicol\'as Cabrera, 13-15, 28049 Madrid, Spain}
\author{David Pérez-García} 
\affiliation{Dpto. An\'alisis Matem\'atico y Matem\'atica Aplicada, Universidad Complutense de Madrid, 28040 Madrid, Spain}
\affiliation{Instituto de Ciencias Matem\'aticas, Campus Cantoblanco UAM, C/ Nicol\'as Cabrera, 13-15, 28049 Madrid, Spain}
\author{Norbert Schuch}
\affiliation{Max-Planck-Institut für Quantenoptik, Hans-Kopfermann-Str. 1, D-85748 Garching, Germany}
\author{J. Ignacio Cirac}
\affiliation{Max-Planck-Institut für Quantenoptik, Hans-Kopfermann-Str. 1, D-85748 Garching, Germany}

\begin{abstract}
  Tensor networks are generated by a set of small rank tensors and define many-body quantum  states in a succinct form. The corresponding map is not one-to-one: different sets of tensors may  generate the very same state. A fundamental question in the study of tensor networks naturally  arises: what is then the relation between those sets? The answer to this question in one dimensional   setups has found several applications, like the characterization of local and global symmetries,  the classification of phases of matter and unitary evolutions, or the determination of the fixed  points of renormalization procedures. Here we answer this question  for projected entangled-pair states (PEPS) in any dimension and lattice geometry, as long as the tensors generating the states are normal, which constitute an important and generic class.
\end{abstract}

\maketitle

\section{Introduction}

Tensor Networks (TNs) provide us with very efficient ways of describing quantum states in discrete systems. They are particularly useful to describe ground \cite{Hastings2006a} and thermal equilibrium states \cite{Hastings2006a,Molnar2015} of local Hamiltonians, or to describe exotic phases of matter \cite{Affleck1988,Schuch2010}. The most prominent examples are matrix product states \cite{Fannes1992,Kluemper1993} (MPS), which portray one-dimensional systems, and their higher dimensional generalization, projected entangled-pair states \cite{Verstraete2003} (PEPS). Their simplicity and special properties makes them very practical in numerical computations \cite{White1992,Verstraete2008,Orus2013,Vidal2007}, as well as in the characterization and classification of a variety of scenarios and phenomena. This includes, for instance, the characterization of symmetry protected phases in spin \cite{Pollmann2010,Chen2010a,Schuch2011} and fermionic \cite{Fidkowski2011} chains, or topological order \cite{Schuch2010,Levin2005,Sahinoglu2014} in two dimensions, lattice gauge theories \cite{Rico2013,Haegeman2014}, unitary evolutions \cite{Schumacher2004,Cirac2017a}, one-way quantum computing \cite{Raussendorf2001,Verstraete2003,Gross2007}, or quantum tomography \cite{Cramer2011}.

Tensor network states can be defined on arbitrary lattices. They are generated by a
set of tensors, $\{A_n\}$, which are assigned to each vertex and are contracted according to the geometry of the lattice. For regular lattices, the generated states are translationally invariant (TI) if all the tensors are the same. A key feature of general TNs is that two different sets of tensors may generate the same tensor network state. This occurs,
for instance, when they are related by a (so-called) gauge transformation; that is, when the tensors of one set are related to the other by matrix multiplication of the indices that are contracted, so that those matrices
cancel with each other once they are contracted. Let us illustrate this with MPS. There, the tensors $A_n$ have rank three: one of the indices corresponds to the physical index, and the other two to the virtual ones that
are contracted in order to generate the state. For a given value of the physical index, $i$, the tensors are just matrices, $A^i_n$. Obviously, the tensors $B_n$, with $B^i_n = X_n A^i_n X_{n+1}^{-1},$ generate the same state
as the tensors $A_n$, where $X_n$ are arbitrary non-singular matrices. One of the fundamental questions in the description of TNs is precisely if this is the only thing that can happen. That is, if two sets of tensors generate the same state, must they be related by a gauge transformation? This question is crucial in many of the applications of tensor networks. For instance, when the answer is affirmative, it gives rise to a canonical form of describing MPS \cite{Fannes1992,Perez-Garcia2007,Vidal2003}. Or, more importantly, it characterizes the tensors
generating states with certain global or local (gauge) symmetries \cite{Sanz2009,Kull2017}. The reason is very simple: if a state is symmetric it means that an operation leaves it invariant; however, in general, it will change the tensors, so that the resulting ones should be related to the original ones by a gauge transformation. This implies that symmetries in the quantum states can be captured by symmetries in the tensors. This question is also decisive in many other situations dealing with string order \cite{Perez-Garcia2008a}, topological order \cite{Sahinoglu2014}, renormalization \cite{Cirac2017}, or time evolution \cite{Cirac2017a}. Theorems answering such fundamental questions about the structure of TNs are typically referred to as Fundamental Theorems.

Proving a Fundamental Theorem for the most general TN is impossible: even for two tensors generating translationally invariant 2D PEPS in an $N\times N$ lattice, there cannot exist an algorithm to decide whether they will generate the same state for all $N$ or not \cite{Scarpa2018}. It is therefore necessary to impose restrictions to the TN (both on the geometry of the lattice as well as on the properties of the defining tensors).
So far, most of the Fundamental Theorems concern MPS. They have been proven for translationally invariant states \cite{Cirac2017,Cuevas2017} as long as the two tensors generate the same state for any size of the lattice. They have also been proven for not necessarily translationally invariant states for a fixed (but large enough) system size for a restricted class of tensors \cite{Perez-Garcia2010}. This class includes injective tensors, that can be inverted by just acting on the physical index, i.e. there exists another tensor, $A^{-1}$, such that
\begin{equation*}
 \sum_i A^i_{\alpha,\beta} (A^{-1})^{i}_{\alpha',\beta'}=
\delta_{\alpha,\alpha'}\delta_{\beta,\beta'},
\end{equation*}
as well as normal tensors, that become injective after blocking a few sites. For 2D PEPS such theorems only exist for restricted (but generic) classes of tensors: for normal tensors \cite{Perez-Garcia2010} and semi-injective tensors \cite{Molnar2017}. These theorems require only a fixed (but large enough) system size. The proof techniques, however, exploit the lattice structure in a fundamental way and thus do not generalize to other geometries.  

In this paper we prove the Fundamental Theorem for normal (and thus also injective) PEPS
in arbitrary lattices (geometries and dimensions). We obtain that if two sets of such tensors
generate the same state, then they must be related by a gauge transformation. This generalizes
the previous results as follows. First, we relax the condition of an existence of a sequence
of TNs (required in e.g. Ref. \onlinecite{Cirac2017}) so that our results hold for a fixed (but large enough) size. The required system size is smaller than in Ref. \onlinecite{Perez-Garcia2010}. Second, the TNs considered
here do not need to be translationally invariant, which is important when applying the results
to local gauge symmetries. Third, the results hold for any geometry (including, for instance,
three dimensions or hyperbolic, as it is used in the constructions of AdS/CFT correspondence
\cite{Pastawski2015,Hayden2016}). Additionally, we show that if a TI PEPS defined in a regular lattice is normal although the tensors are different in different sites, then there exists a TI PEPS description with the same bond dimension and where the tensors at every site are the same.
Furthermore, the proof presented here uses a new technique: even though it relies on a reduction to the MPS case, this reduction is done in a local way instead of ``slicing'' a PEPS into an MPS along one dimension.

\section{Injective MPS}

In this Section we define non-translational invariant \emph{injective} MPS. We show that two such MPS generate the same state if and only if the generating tensors are related with a gauge transformation (if the MPS contains at least three sites). This extends the previously known results as here we consider (i) a fixed system size and (ii) non-translational invariant MPS with closed boundary conditions. 

A non translational invariant MPS on $n$ particles is a state
\begin{equation*}
  \ket{\Psi} = \sum_i \tr\{A_1^{i_1} A_2^{i_2} \dots A_n^{i_n} \} \ket{i_1 \dots i_n},
\end{equation*}
where each $i_k$ runs through a basis of the (finite dimensional) Hilbert space associated to the $k$th particle and each $A_k^{i_k}$ is a  $D_k \times D_{k+1}$ matrix ($D_{n+1}=D_1$). From now on, we will use graphical notation: each tensor is depicted by a dot with lines attached to it. The lines correspond to the different indices of the tensor; joining the lines correspond to contraction of indices. For example, a scalar is represented by a single dot with no lines joinig to it, a vector is represented by a dot with a single line attached to it, a matrix by a dot with two lines attached to it:
\begin{equation*}
  s = \begin{tikzpicture}
    \node[tensor, label=below:$s$] {};
  \end{tikzpicture} \ , \quad
  \ket{v} = \begin{tikzpicture}
    \node[tensor, label=below:$v$] (v) {};
    \draw (v)--++(-0.5,0);
  \end{tikzpicture} \ , \quad 
  A = \begin{tikzpicture}
    \node[tensor, label=below:$A$] (A) {};
    \draw (-0.5,0)--(0.5,0);
  \end{tikzpicture} \ ;  
\end{equation*}
the scalar product of two vectors, the action of a matrix on a vector and a matrix element can be written as
\begin{equation*}
  \langle w \vert v \rangle = \begin{tikzpicture}
    \node[tensor, label=below:$w$] at (0,0) (v) {};
    \node[tensor, label=below:$v$] at (1,0) (w) {};
    \draw (v)--(w);
  \end{tikzpicture} \ , \quad 
  A \ket{v} = \begin{tikzpicture}
    \node[tensor, label=below:$A$] at (0,0) (A) {};
    \node[tensor, label=below:$v$] at (1,0) (v) {};
    \draw (-0.5,0)--(v);
  \end{tikzpicture} \ , \quad 
 \langle w \vert A \ket{v} = \begin{tikzpicture}
    \node[tensor, label=below:$A$] at (0,0) (A) {};
    \node[tensor, label=below:$w$] at (-1,0) (v) {};
    \node[tensor, label=below:$v$] at (1,0) (w) {};
    \draw (v)--(w);
 \end{tikzpicture} \ .
\end{equation*}
In this notation, the MPS $\ket{\Psi}$ is written as
\begin{equation*}
  \ket{\Psi} = \ 
  \begin{tikzpicture}
    \draw (0.5,0) rectangle (5.5,-0.5);
    \foreach \x/\t in {1/1,2/2,3/3,5/n}{
      \node[tensor,label=below:$A_\t$] (t\x) at (\x,0) {};
      \draw[very thick] (t\x) --++ (0,0.5);
    }
    \node[fill=white] at (4,0) {$\dots$};
  \end{tikzpicture}  \ .
\end{equation*}  
An \emph{injective} MPS is an MPS where every tensor -- if considered as a map from the virtual level to the physical one -- is injective, i.e.\ 
\begin{equation*}
  \begin{tikzpicture}
    \draw (-0.5,0) rectangle (0.5,-0.6);
    \node[tensor,label=below:$A_i$] (a) at (0,0) {};
    \draw[very thick] (a)--++(0,0.5);
    \node[tensor,label=below:$X$] at (0,-0.6) {};
  \end{tikzpicture} \ = 0 \quad \Rightarrow \quad X = 0.
\end{equation*}
This is equivalent to the tensor $A_i$ admitting a one-sided inverse $A_i^{-1}$:
\begin{equation*}
  \begin{tikzpicture}
    \draw (-0.5,0)--(0.5,0);
    \draw (-0.5,0.5)--(0.5,0.5);
    \draw[very thick] (0,0)--(0,0.5);
    \node[tensor,label=below:$A_i$] at (0,0) {};
    \node[tensor,label=above:$A_i^{-1}$] at (0,0.5) {};
  \end{tikzpicture}  =
  \begin{tikzpicture}
    \begin{scope}
      \clip (0.2,-0.2) rectangle (1,0.7);
      \draw (0,0) rectangle (0.5,0.5);
      \draw[shift={(0.7,0)}] (0,0) rectangle (0.5,0.5);
    \end{scope}
  \end{tikzpicture} \ . 
\end{equation*}  
Notice that this immediately shows that the contraction of two injective MPS tensors is again injective; the inverse of the obtained tensor is proportional to the contraction of the inverses of the individual tensors:
\begin{equation*}
  \begin{tikzpicture}
    \draw (-0.5,0)--(1.5,0);
    \draw (-0.5,0.5)--(1.5,0.5);
    \draw[very thick] (0,0)--++(0,0.5);
    \draw[very thick] (1,0)--++(0,0.5);
    \node[tensor,label=below:$A_1$] at (0,0) {};
    \node[tensor,label=above:$A_1^{-1}$] at (0,0.5) {};
    \node[tensor,label=below:$A_2$] at (1,0) {};
    \node[tensor,label=above:$A_2^{-1}$] at (1,0.5) {};
  \end{tikzpicture}  =
  \begin{tikzpicture}
    \begin{scope}
      \clip (0.2,-0.2) rectangle (1.65,0.7);
      \draw (0,0) rectangle (0.5,0.5);
      \draw[shift={(0.7,0)}] (0,0) rectangle (0.5,0.5);
      \draw[shift={(1.4,0)}] (0,0) rectangle (0.5,0.5);
    \end{scope}
  \end{tikzpicture} \ = D \cdot \ 
  \begin{tikzpicture}
    \begin{scope}
      \clip (0.2,-0.2) rectangle (1,0.7);
      \draw (0,0) rectangle (0.5,0.5);
      \draw[shift={(0.7,0)}] (0,0) rectangle (0.5,0.5);
    \end{scope}
  \end{tikzpicture} \ , 
\end{equation*}
where $D$ is the dimension of the vector space assigned to the index connecting the tensors $A_1$ and $A_2$. 

In the rest of this Section, we prove the two main lemmas leading to the Fundamental Theorem. We also illustrate how to use them by deriving the Fundamental Theorem for non translational invariant MPS. In the following, we consider two injective tensor networks generating the same state; the defining tensors of the two TNs are labeled by $A_i$ and $B_i$. The first lemma assigns a special gauge transformation to each edge of one of the tensor networks; the second lemma shows that once these gauges are absorbed into the tensors $B_i$, the resulting tensors are equal to $A_i$.

\begin{lemma}\label{lem:inj_isomorph}
  Suppose $A,B$ are two injective, non translational invariant MPS on three sites that generate the same state. Then for every edge and for every matrix $X$  there is a matrix $Y$ such that
  \begin{equation*}
    \begin{tikzpicture}
      \draw (0.5,0) rectangle (3.5,-0.5);
      \foreach \x in {1,2,3}{
        \node[tensor,label=below:$A_\x$] (t\x) at (\x,0) {};
        \draw[very thick] (t\x) --++ (0,0.5);
      }
      \node[tensor, red,label=above:$X$] (x) at (1.5,0) {};
    \end{tikzpicture} = 
    \begin{tikzpicture}
      \draw (0.5,0) rectangle (3.5,-0.5);
      \foreach \x in {1,2,3}{
        \node[tensor,label=below:$B_\x$] (t\x) at (\x,0) {};
        \draw[very thick] (t\x) --++ (0,0.5);
      }
      \node[tensor, red,label=above:$Y$] (x) at (1.5,0) {};
    \end{tikzpicture} \ .
  \end{equation*} 
  Moreover, $X$ and $Y$ have the same dimension and there is an invertible matrix $Z$ such that $Y=Z^{-1}X Z$. This $Z$ is uniquely defined up to multiplication with a constant.
\end{lemma}

This Lemma will be used to assign a local gauge transformation to all edges on one of two tensor networks generating the same state. These local gauges will then be incorporated into the defining tensors; doing so will lead to two tensor networks where inserting any matrix $X$ on any bond simultaneously in the two networks gives two new states that are still equal. 

The proof of \cref{lem:inj_isomorph} is based on the observation that any local operation on the virtual level can be realized by a physical one on either of the neighboring particles; and vice versa, two physical operations on neighboring particles that transform the state the same way correspond to a virtual operation on the bond connecting the two particles. Given two tensor networks generating the same state, this correspondence establishes an isomorphism between the algebra of virtual operations. The basis change realizing this isomorphism is the local gauge relating the two tensors.

Before proceeding to the proof, notice that due to injectivity of the tensors, if 
\begin{equation*}
  \begin{tikzpicture}
    \draw (0.5,0) rectangle (3.5,-0.5);
    \foreach \x in {1,2,3}{
      \node[tensor,label=below:$A_\x$] (t\x) at (\x,0) {};
      \draw[very thick] (t\x) --++ (0,0.5);
    }
    \node[tensor, red,label=above:$X_1$] (x) at (1.5,0) {};
  \end{tikzpicture} = 
  \begin{tikzpicture}
    \draw (0.5,0) rectangle (3.5,-0.5);
    \foreach \x in {1,2,3}{
      \node[tensor,label=below:$A_\x$] (t\x) at (\x,0) {};
      \draw[very thick] (t\x) --++ (0,0.5);
    }
    \node[tensor, red,label=above:$X_2$] (x) at (1.5,0) {};
  \end{tikzpicture} \ ,
\end{equation*}
then $X_1=X_2$.

\begin{proof}[Proof of \cref{lem:inj_isomorph}]
Consider now a deformation of the TN by inserting a matrix $X$ on one of the bonds. This deformation can be realized by physical operations acting on either of the two neighboring particles:
\begin{equation*}
  \begin{tikzpicture}[baseline=-0.1cm]
    \draw (0.5,0) rectangle (3.5,-0.5);
    \foreach \x in {1,2,3}{
      \node[tensor,label=below:$A_\x$] (t\x) at (\x,0) {};
      \draw[very thick] (t\x) --++ (0,0.5);
    }
    \node[tensor, red,label=above:$X$] (x) at (1.5,0) {};
  \end{tikzpicture} = 
  \begin{tikzpicture}[baseline=-0.1cm]
    \draw (0.5,0) rectangle (3.5,-0.5);
    \foreach \x in {1,2,3}{
      \node[tensor,label=below:$A_\x$] (t\x) at (\x,0) {};
      \draw[very thick] (t\x) --++ (0,0.5);
    }
    \node[tensor, red,label=left:$O_1$] (o) at (1,0.5) {};
    \draw[very thick,red] (o)--++(0,0.5); 
  \end{tikzpicture} = 
  \begin{tikzpicture}[baseline=-0.1cm]
    \draw (0.5,0) rectangle (3.5,-0.5);
    \foreach \x in {1,2,3}{
      \node[tensor,label=below:$A_\x$] (t\x) at (\x,0) {};
      \draw[very thick] (t\x) --++ (0,0.5);
    }
    \node[tensor, red,label=left:$O_2$] (o) at (2,0.5) {};
    \draw[very thick,red] (o)--++(0,0.5); 
  \end{tikzpicture} \ ,
\end{equation*} 
with 
\begin{equation}\label{eq:X->O}
   O_1 = 
  \begin{tikzpicture}
    \draw (-0.5,0) rectangle (0.5,1.2);
    \node[tensor,label=below:$A_1$] (t) at (0,1.2) {};
    \node[tensor,label=above:$A_1^{-1}$] (b) at (0,0) {};
    \node[tensor,label=right:$X$] (x) at (0.5,0.6) {};
    \draw[very thick] (t)--++(0,0.5);
    \draw[very thick] (b)--++(0,-0.5);
  \end{tikzpicture} \quad \text{and} \quad 
   O_2 = 
  \begin{tikzpicture}
    \draw (-0.5,0) rectangle (0.5,1.2);
    \node[tensor,label=below:$A_2$] (t) at (0,1.2) {};
    \node[tensor,label=above:$A_2^{-1}$] (b) at (0,0) {};
    \node[tensor,label=left:$X$] (x) at (-0.5,0.6) {};
    \draw[very thick] (t)--++(0,0.5);
    \draw[very thick] (b)--++(0,-0.5);
  \end{tikzpicture} \ .     
\end{equation}
Notice that the mappings $X\mapsto O_1$ and $X\mapsto O_2^T$ are algebra homomorphisms\footnote{The virtual bonds of the tensors $A_i$ read from left to right, thus the loops in \cref{eq:O->X} read clockwise; hence the transpose in the mapping $X\mapsto O_2^T$.}. These mappings do not depend on $A_3$.

Consider now the converse: two physical operations on neighboring particles that transform the MPS to the same state:
\begin{equation}\label{eq:resonate}
  \begin{tikzpicture}[baseline=-0.1cm]
    \draw (0.5,0) rectangle (3.5,-0.5);
    \foreach \x in {1,2,3}{
      \node[tensor,label=below:$B_\x$] (t\x) at (\x,0) {};
      \draw[very thick] (t\x) --++ (0,0.5);
    }
    \node[tensor, red,label=left:$O_1$] (o) at (1,0.5) {};
    \draw[very thick,red] (o)--++(0,0.5); 
  \end{tikzpicture} = 
  \begin{tikzpicture}[baseline=-0.1cm]
    \draw (0.5,0) rectangle (3.5,-0.5);
    \foreach \x in {1,2,3}{
      \node[tensor,label=below:$B_\x$] (t\x) at (\x,0) {};
      \draw[very thick] (t\x) --++ (0,0.5);
    }
    \node[tensor, red,label=left:$O_2$] (o) at (2,0.5) {};
    \draw[very thick,red] (o)--++(0,0.5); 
  \end{tikzpicture} \ .
\end{equation}
Inverting $B_2$ and $B_3$, we arrive to 
\begin{equation}\label{eq:inj_O->X_argument}
  \begin{tikzpicture}
    \node[tensor,label=below:$B_1$] (l) at (0,-0.6) {};
      \draw[very thick] (l)--++(0,0.6);
    \draw (l)++(-0.6,0)--(l)--++(0.6,0);
    \node[tensor,red,label=left:$O_1$] (o) at ($(l)+(0,0.5)$) {};
    \draw[red, very thick] (o) --++(0,0.5);
  \end{tikzpicture} = D_{23}^{-1}
  \begin{tikzpicture}
    \node[tensor,label=below:$B_1$] (l) at (0,-0.6) {};
    \node[tensor,label=below:$B_2$] (r) at (1,-0.6) {};
    \foreach \x in {l,r}{
      \draw[very thick] (\x)--++(0,0.6);
    }
    \draw (l)++(-0.6,0)--(l)--(r)--++(0.6,0);
    \node[tensor,red,label=left:$O_2$] (o) at ($(r)+(0,0.5)$) {};
    \draw[red, very thick] (o) --++(0,0.5);      
    \node[tensor,label=above:$B_2^{-1}$] (i) at ($ (r) + (0,1)$) {};
    \draw (i)++(-0.6,0)--(i)--++(0.6,0)--++(0,-1);
  \end{tikzpicture}  =  
  \begin{tikzpicture}
    \draw (-0.5,0)--(1.5,0);
    \node[tensor,label=below:$B_1$] (a) at (0,0) {};
    \draw[very thick] (a)--++(0,0.5);
    \node[tensor, label=below:$W$] (x) at (1,0) {};
  \end{tikzpicture} \ ,
\end{equation}  
for some matrix $W$, where $D_{23}$ is the dimension of the vector space on the edge $(2,3)$.  Similarly, inverting $B_1$ and $B_3$, we arrive to 
\begin{equation*}
  \begin{tikzpicture}
    \draw (-0.5,0)--(0.5,0);
    \node[tensor,label=below:$B_2$] (a) at (0,0) {};
    \draw[very thick] (a)--++(0,0.5);
    \node[tensor, red,label=left:$O_2$] (o) at (0,0.5) {};
    \draw[very thick,red] (o)--++(0,0.5); 
  \end{tikzpicture} = 
  \begin{tikzpicture}
    \draw (-1.5,0)--(0.5,0);
    \node[tensor,label=below:$B_2$] (a) at (0,0) {};
    \draw[very thick] (a)--++(0,0.5);
    \node[tensor, label=below:$V$] (x) at (-1,0) {};
  \end{tikzpicture} \ ,
\end{equation*}
for some matrix $V$.  Therefore
\begin{equation*}
  \begin{tikzpicture}[baseline=-0.1cm]
    \draw (0.5,0) rectangle (3.5,-0.5);
    \foreach \x in {1,2,3}{
      \node[tensor,label=below:$B_\x$] (t\x) at (\x,0) {};
      \draw[very thick] (t\x) --++ (0,0.5);
    }
    \node[tensor, red,label=above:$W$] (x) at (1.5,0) {};
  \end{tikzpicture} = 
  \begin{tikzpicture}[baseline=-0.1cm]
    \draw (0.5,0) rectangle (3.5,-0.5);
    \foreach \x in {1,2,3}{
      \node[tensor,label=below:$B_\x$] (t\x) at (\x,0) {};
      \draw[very thick] (t\x) --++ (0,0.5);
    }
    \node[tensor, red,label=left:$O_1$] (o) at (1,0.5) {};
    \draw[very thick,red] (o)--++(0,0.5); 
  \end{tikzpicture} = 
  \begin{tikzpicture}[baseline=-0.1cm]
    \draw (0.5,0) rectangle (3.5,-0.5);
    \foreach \x in {1,2,3}{
      \node[tensor,label=below:$B_\x$] (t\x) at (\x,0) {};
      \draw[very thick] (t\x) --++ (0,0.5);
    }
    \node[tensor, red,label=left:$O_2$] (o) at (2,0.5) {};
    \draw[very thick,red] (o)--++(0,0.5); 
  \end{tikzpicture} \ = \ 
  \begin{tikzpicture}[baseline=-0.1cm]
    \draw (0.5,0) rectangle (3.5,-0.5);
    \foreach \x in {1,2,3}{
      \node[tensor,label=below:$B_\x$] (t\x) at (\x,0) {};
      \draw[very thick] (t\x) --++ (0,0.5);
    }
    \node[tensor, red,label=above:$V$] (x) at (1.5,0) {};
  \end{tikzpicture} \ ,
\end{equation*}
and thus by injectivity, $V=W$. Therefore
\begin{equation} \label{eq:O->X}
  \begin{tikzpicture}
    \draw (-0.5,0)--(0.5,0);
    \node[tensor,label=below:$B_1$] (a) at (0,0) {};
    \draw[very thick] (a)--++(0,0.5);
    \node[tensor, red,label=left:$O_1$] (o) at (0,0.5) {};
    \draw[very thick,red] (o)--++(0,0.5); 
  \end{tikzpicture} = 
  \begin{tikzpicture}
    \draw (-0.5,0)--(1.5,0);
    \node[tensor,label=below:$B_1$] (a) at (0,0) {};
    \draw[very thick] (a)--++(0,0.5);
    \node[tensor, label=below:$W$] (x) at (1,0) {};
  \end{tikzpicture} \quad \text{and} \quad 
  \begin{tikzpicture}
    \draw (-0.5,0)--(0.5,0);
    \node[tensor,label=below:$B_2$] (a) at (0,0) {};
    \draw[very thick] (a)--++(0,0.5);
    \node[tensor, red,label=left:$O_2$] (o) at (0,0.5) {};
    \draw[very thick,red] (o)--++(0,0.5); 
  \end{tikzpicture} = 
  \begin{tikzpicture}
    \draw (-1.5,0)--(0.5,0);
    \node[tensor,label=below:$B_2$] (a) at (0,0) {};
    \draw[very thick] (a)--++(0,0.5);
    \node[tensor, label=below:$W$] (x) at (-1,0) {};
  \end{tikzpicture} \ ,
\end{equation}
and the maps $O_1\mapsto W$ and $O_2^T \mapsto W$ are uniquely defined and are  algebra homomorphisms. 

Consider now two three-particle, non translational invariant injective MPS generating the same state:
\begin{equation*}
  \begin{tikzpicture}
    \draw (0.5,0) rectangle (3.5,-0.5);
    \foreach \x in {1,2,3}{
      \node[tensor,label=below:$A_\x$] (t\x) at (\x,0) {};
      \draw[very thick] (t\x) --++ (0,0.5);
    }
  \end{tikzpicture} = 
  \begin{tikzpicture}
    \draw (0.5,0) rectangle (3.5,-0.5);
    \foreach \x in {1,2,3}{
      \node[tensor,label=below:$B_\x$] (t\x) at (\x,0) {};
      \draw[very thick] (t\x) --++ (0,0.5);
    }
  \end{tikzpicture} \ .
\end{equation*}
Deform the MPS on the LHS by inserting a matrix $X$ on one of the bonds. By the above arguments, this deformation is equivalent to any of the two physical operations:
\begin{equation*}
  \begin{tikzpicture}[baseline=-0.1cm]
    \draw (0.5,0) rectangle (3.5,-0.5);
    \foreach \x in {1,2,3}{
      \node[tensor,label=below:$A_\x$] (t\x) at (\x,0) {};
      \draw[very thick] (t\x) --++ (0,0.5);
    }
    \node[tensor, red,label=above:$X$] (x) at (1.5,0) {};
  \end{tikzpicture} = 
  \begin{tikzpicture}[baseline=-0.1cm]
    \draw (0.5,0) rectangle (3.5,-0.5);
    \foreach \x in {1,2,3}{
      \node[tensor,label=below:$A_\x$] (t\x) at (\x,0) {};
      \draw[very thick] (t\x) --++ (0,0.5);
    }
    \node[tensor, red,label=left:$O_1$] (o) at (1,0.5) {};
    \draw[very thick,red] (o)--++(0,0.5); 
  \end{tikzpicture} = 
  \begin{tikzpicture}[baseline=-0.1cm]
    \draw (0.5,0) rectangle (3.5,-0.5);
    \foreach \x in {1,2,3}{
      \node[tensor,label=below:$A_\x$] (t\x) at (\x,0) {};
      \draw[very thick] (t\x) --++ (0,0.5);
    }
    \node[tensor, red,label=left:$O_2$] (o) at (2,0.5) {};
    \draw[very thick,red] (o)--++(0,0.5); 
  \end{tikzpicture} \ .
\end{equation*} 
As the MPS defined by the $A$ and $B$ tensors is the same state, these physical operators also satisfy
\begin{equation*}
  \begin{tikzpicture}[baseline=-0.1cm]
    \draw (0.5,0) rectangle (3.5,-0.5);
    \foreach \x in {1,2,3}{
      \node[tensor,label=below:$A_\x$] (t\x) at (\x,0) {};
      \draw[very thick] (t\x) --++ (0,0.5);
    }
    \node[tensor, red,label=above:$X$] (x) at (1.5,0) {};
  \end{tikzpicture} = 
  \begin{tikzpicture}[baseline=-0.1cm]
    \draw (0.5,0) rectangle (3.5,-0.5);
    \foreach \x in {1,2,3}{
      \node[tensor,label=below:$B_\x$] (t\x) at (\x,0) {};
      \draw[very thick] (t\x) --++ (0,0.5);
    }
    \node[tensor, red,label=left:$O_1$] (o) at (1,0.5) {};
    \draw[very thick,red] (o)--++(0,0.5); 
  \end{tikzpicture} = 
  \begin{tikzpicture}[baseline=-0.1cm]
    \draw (0.5,0) rectangle (3.5,-0.5);
    \foreach \x in {1,2,3}{
      \node[tensor,label=below:$B_\x$] (t\x) at (\x,0) {};
      \draw[very thick] (t\x) --++ (0,0.5);
    }
    \node[tensor, red,label=left:$O_2$] (o) at (2,0.5) {};
    \draw[very thick,red] (o)--++(0,0.5); 
  \end{tikzpicture} \ ,
\end{equation*} 
and thus, by \cref{eq:O->X}, for every $X$ there is a matrix $Y$ such that 
\begin{equation*}
  \begin{tikzpicture}
    \draw (0.5,0) rectangle (3.5,-0.5);
    \foreach \x in {1,2,3}{
      \node[tensor,label=below:$A_\x$] (t\x) at (\x,0) {};
      \draw[very thick] (t\x) --++ (0,0.5);
    }
    \node[tensor, red,label=above:$X$] (x) at (1.5,0) {};
  \end{tikzpicture} = 
  \begin{tikzpicture}
    \draw (0.5,0) rectangle (3.5,-0.5);
    \foreach \x in {1,2,3}{
      \node[tensor,label=below:$B_\x$] (t\x) at (\x,0) {};
      \draw[very thick] (t\x) --++ (0,0.5);
    }
    \node[tensor, red,label=above:$Y$] (x) at (1.5,0) {};
  \end{tikzpicture} \ .
\end{equation*}
Due to injectivity of the $B$ tensors, the mapping $X\mapsto Y$ is uniquely defined. Due to injectivity of the $A$ tensors, it is an injective map. As the argument is symmetric with respect of the exchange of the $A$ and $B$ tensors, it also has to be surjective and therefore the map $X\mapsto Y$ is a bijection. Moreover, it is clear from the construction that it is an algebra homomorphism, as both $X\mapsto O_1$ and $O_1\mapsto Y$ are algebra homomorphisms. Therefore the mapping $X\mapsto Y$ is an algebra isomorphism. As $X$ (and $Y$) can be any matrix on the bond, this means that the bond dimensions on the LHS and the RHS are the same and that $Y=ZXZ^{-1}$ for some invertible $Z$ and this $Z$ is uniquely defined (up to a multiplicative constant).
\end{proof}

\begin{lemma}\label{lem:inj_equal_tensors}
  Let $A_1,A_2$ and $B_1,B_2$ be injective MPS tensors. Suppose that for all $X$ and $Y$
  \begin{equation*}
    \begin{tikzpicture}
      \draw (0.5,0) rectangle (2.5,-0.5);
      \foreach \x in {1,2}{
        \node[tensor,label=below:$A_\x$] (t\x) at (\x,0) {};
        \draw[very thick] (t\x) --++ (0,0.5);
      }
      \node[tensor, red,label=above:$X$] (x) at (1.5,0) {};
    \end{tikzpicture} = 
    \begin{tikzpicture}
      \draw (0.5,0) rectangle (2.5,-0.5);
      \foreach \x in {1,2}{
        \node[tensor,label=below:$B_\x$] (t\x) at (\x,0) {};
        \draw[very thick] (t\x) --++ (0,0.5);
      }
      \node[tensor, red,label=above:$X$] (x) at (1.5,0) {};
    \end{tikzpicture} \quad \text{and} \quad
    \begin{tikzpicture}
      \draw (0.5,0) rectangle (2.5,-0.5);
      \foreach \x in {1,2}{
        \node[tensor,label=below:$A_\x$] (t\x) at (\x,0) {};
        \draw[very thick] (t\x) --++ (0,0.5);
      }
      \node[tensor, red,label=below:$Y$] (x) at (1.5,-0.5) {};
    \end{tikzpicture} = 
    \begin{tikzpicture}
      \draw (0.5,0) rectangle (2.5,-0.5);
      \foreach \x in {1,2}{
        \node[tensor,label=below:$B_\x$] (t\x) at (\x,0) {};
        \draw[very thick] (t\x) --++ (0,0.5);
      }
      \node[tensor, red,label=below:$Y$] (x) at (1.5,-0.5) {};
    \end{tikzpicture} \ .    
  \end{equation*}
  Then $A_1 = \lambda B_1$ and $A_2 = \lambda^{-1} B_2$ for some constant $\lambda$.
\end{lemma}

\begin{proof}
  From the first equation, as $X$ can be any matrix, 
  \begin{equation*}
    \begin{tikzpicture}
      \draw (0.5,0)--(2.5,0);
      \node[tensor,label=below:$A_2$] (t1) at (1,0) {};
      \node[tensor,label=below:$A_1$] (t2) at (2,0) {};
      \draw[very thick] (t1) --++ (0,0.5);
      \draw[very thick] (t2) --++ (0,0.5);
    \end{tikzpicture} \ = \ 
    \begin{tikzpicture}
      \draw (0.5,0)--(2.5,0);
      \node[tensor,label=below:$B_2$] (t1) at (1,0) {};
      \node[tensor,label=below:$B_1$] (t2) at (2,0) {};
      \draw[very thick] (t1) --++ (0,0.5);
      \draw[very thick] (t2) --++ (0,0.5);
    \end{tikzpicture} \ .
  \end{equation*}
  Similarly, from the second equation,
  \begin{equation*}
    \begin{tikzpicture}
      \draw (0.5,0)--(2.5,0);
      \node[tensor,label=below:$A_1$] at (1,0) {};
      \node[tensor,label=below:$A_2$] at (2,0) {};
      \draw[very thick] (t1) --++ (0,0.5);
      \draw[very thick] (t2) --++ (0,0.5);
    \end{tikzpicture} \ = \ 
    \begin{tikzpicture}
      \draw (0.5,0)--(2.5,0);
      \node[tensor,label=below:$B_1$] at (1,0) {};
      \node[tensor,label=below:$B_2$] at (2,0) {};
      \draw[very thick] (t1) --++ (0,0.5);
      \draw[very thick] (t2) --++ (0,0.5);
    \end{tikzpicture} \ .
  \end{equation*}  
  Therefore, applying $A_2^{-1}$ to both equations, we get that
  \begin{equation*}
    \begin{tikzpicture}
      \draw (-0.5,0)--(0.5,0);
      \node[tensor,label=below:$A_1$] (a) at (0,0) {};
      \draw[very thick] (a)--++(0,0.5);
    \end{tikzpicture} \ = \ 
    \begin{tikzpicture}
      \draw (-0.5,0)--(1.5,0);
      \node[tensor,label=below:$B_1$] (a) at (0,0) {};
      \draw[very thick] (a)--++(0,0.5);
      \node[tensor, label=below:$Z$] (x) at (1,0) {};
    \end{tikzpicture} \ = \ 
    \begin{tikzpicture}
      \draw (-0.5,0)--(1.5,0);
      \node[tensor,label=below:$B_1$] (a) at (1,0) {};
      \draw[very thick] (a)--++(0,0.5);
      \node[tensor, label=below:$W$] (x) at (0,0) {};
    \end{tikzpicture} \ ,
  \end{equation*}
  for some matrices $Z$ and $W$. Applying the inverse of $B_1$, we conclude that both $Z$ and $W$ are proportional to identity and hence $A_1 = \lambda B_1$. Similarly $A_2 = \mu B_2$ for some other constant $\mu$ and $\mu = 1/\lambda$.
\end{proof}

In the following, we show how to use these lemmas for injective MPS to prove the Fundamental Theorem. This is a special case of the next section, and only presented to explain the ideas.

\begin{theorem}\label{thm:inj_MPS}
  Let the tensors $A_i$ and $B_i$ define two injective, non translational invariant MPS on at least three particles. Suppose they generate the same state:
  \begin{equation*}
    \ket{\Psi} = 
    \begin{tikzpicture}
      \draw (0.5,0) rectangle (4.5,-0.5);
      \foreach \x/\t in {1/1,2/2,4/n}{
        \node[tensor,label=below:$A_\t$] (t\x) at (\x,0) {};
        \draw[very thick] (t\x) --++ (0,0.5);
      }
      \node[fill=white] at (3,0) {$\dots$};
    \end{tikzpicture} = 
    \begin{tikzpicture}
      \draw (0.5,0) rectangle (4.5,-0.5);
      \foreach \x/\t in {1/1,2/2,4/n}{
        \node[tensor,label=below:$B_\t$] (t\x) at (\x,0) {};
        \draw[very thick] (t\x) --++ (0,0.5);
      }
      \node[fill=white] at (3,0) {$\dots$};
    \end{tikzpicture} \ .
  \end{equation*}  
  Then there are invertible matrices $Z_i$ ($i=1,...,n+1$, $Z_{i+1}=Z_1$) such that 
    \begin{equation*}
      \begin{tikzpicture}
        \draw (-0.5,0)--(0.5,0);
        \node[tensor,label=below:$B_i$] (t) at (0,0) {};
        \draw[very thick] (t)--(0,0.5);
      \end{tikzpicture}  = 
      \begin{tikzpicture}
        \draw (-1,0)--(1,0);
        \node[tensor,label=below:$Z_i^{-1}$] at (-0.5,0) {};
        \node[tensor,label=below:\ $Z_{i+1}$\vphantom{$Z_i^{-1}$}] at (0.5,0) {};
        \node[tensor,label=below:$A_i$\vphantom{$Z_i^{-1}$}] (t) at (0,0) {};
        \draw[very thick] (t)--(0,0.5);
      \end{tikzpicture} \ .
    \end{equation*}  
    Moreover, the gauges $Z_i$ are unique up to a multiplicative constant.
\end{theorem}

\begin{proof}
  First let us choose any edge, for example the edge $(1,2)$. Let us block the tensors $A_3,\dots A_n$ (and $B_3,\dots B_n$) into one tensor:
  \begin{align*}
    \begin{tikzpicture}
      \draw (0.5,0) -- (1.5,0);
      \node[tensor,label=below:$a$] (t) at (1,0) {};
      \draw[very thick] (t) --++ (0,0.5);
    \end{tikzpicture} &= 
    \begin{tikzpicture}
      \draw (2.5,0) -- (6.5,0);
      \foreach \x/\t in {3/3,4/4,6/n}{
        \node[tensor,label=below:$A_\t$] (t\x) at (\x,0) {};
        \draw[very thick] (t\x) --++ (0,0.5);
      }
      \node[fill=white] at (5,0) {$\dots$};
    \end{tikzpicture} \\
    \begin{tikzpicture}
      \draw (0.5,0) -- (1.5,0);
      \node[tensor,label=below:$b$] (t) at (1,0) {};
      \draw[very thick] (t) --++ (0,0.5);
    \end{tikzpicture} &= 
    \begin{tikzpicture}
      \draw (2.5,0) -- (6.5,0);
      \foreach \x/\t in {3/3,4/4,6/n}{
        \node[tensor,label=below:$B_\t$] (t\x) at (\x,0) {};
        \draw[very thick] (t\x) --++ (0,0.5);
      }
      \node[fill=white] at (5,0) {$\dots$};
    \end{tikzpicture} \ .
  \end{align*}
  As injectivity is preserved under blocking, both $a$ and $b$ are injective tensors. With this notation, the MPS can be written as a non translational invariant MPS on three sites:
  \begin{equation*}
    \ket{\Psi} = 
    \begin{tikzpicture}
      \draw (0.5,0) rectangle (3.5,-0.5);
      \foreach \x/\t in {1/A_1,2/A_2,3/a}{
        \node[tensor,label=below:$\t$] (t\x) at (\x,0) {};
        \draw[very thick] (t\x) --++ (0,0.5);
      }
    \end{tikzpicture} = 
    \begin{tikzpicture}
      \draw (0.5,0) rectangle (3.5,-0.5);
      \foreach \x/\t in {1/B_1,2/B_2,3/b}{
        \node[tensor,label=below:$\t$] (t\x) at (\x,0) {};
        \draw[very thick] (t\x) --++ (0,0.5);
      }
    \end{tikzpicture} \ .
  \end{equation*}    
  Therefore \cref{lem:inj_isomorph} can be applied leading to a gauge transform $Z_2$ on the edge $(1,2)$ that, for all $X$ with  $Y = Z_2^{-1} X Z_2$, satisfies
  \begin{equation*}
    \begin{tikzpicture}
      \draw (0.5,0) rectangle (3.5,-0.5);
      \foreach \x/\t in {1/$A_1$,2/$A_2$,3/$a$}{
        \node[tensor,label=below:\t] (t\x) at (\x,0) {};
        \draw[very thick] (t\x) --++ (0,0.5);
      }
      \node[tensor, red,label=above:$X$] (x) at (1.5,0) {};
    \end{tikzpicture} = 
    \begin{tikzpicture}
      \draw (0.5,0) rectangle (3.5,-0.5);
      \foreach \x/\t in {1/$B_1$,2/$B_2$,3/$b$}{
        \node[tensor,label=below:\t] (t\x) at (\x,0) {};
        \draw[very thick] (t\x) --++ (0,0.5);
      }
      \node[tensor, red,label=above:$Y$] (x) at (1.5,0) {};
    \end{tikzpicture} \ .
  \end{equation*}
  The lemma can be applied to all edges leading to gauge $Z_i$ on the edge $(i-1,i)$. After incorporating these gauges into the tensor $B_i$:
  \begin{equation}\label{eq:B tilde}
    \begin{tikzpicture}
      \draw (-0.5,0)--(0.5,0);
      \node[tensor,label=below:$\tilde{B}_i$] (t) at (0,0) {};
      \draw[very thick] (t)--(0,0.5);
    \end{tikzpicture}  = 
    \begin{tikzpicture}
      \draw (-1,0)--(1,0);
      \node[tensor,label=below:$Z_i$\vphantom{$Z_{i+1}^{-1}$}] at (-0.5,0) {};
      \node[tensor,label=below:\ $Z_{i+1}^{-1}$] at (0.5,0) {};
      \node[tensor,label=below:$B_i$\vphantom{$Z_{i+1}^{-1}$}] (t) at (0,0) {};
      \draw[very thick] (t)--(0,0.5);
    \end{tikzpicture} \ ,
  \end{equation}  
  we arrive at two MPS with the property that on every bond for every matrix $X$
  \begin{equation*} 
    \begin{tikzpicture}
      \draw (0.5,0) rectangle (4.5,-0.6);
      \foreach \x/\t in {1/1,2/2,4/n}{
        \node[tensor,label=below:$A_\t$] (t\x) at (\x,0) {};
        \draw[very thick] (t\x) --++ (0,0.5);
      }
      \node[fill=white] at (3,0) {$\dots$};
      \node[tensor, red,label=above:$X$] (x) at (1.5,0) {};
    \end{tikzpicture} = 
    \begin{tikzpicture}
      \draw (0.5,0) rectangle (4.5,-0.6);
      \foreach \x/\t in {1/1,2/2,4/n}{
        \node[tensor,label=below:$\tilde{B}_\t$] (t\x) at (\x,0) {};
        \draw[very thick] (t\x) --++ (0,0.5);
      }
      \node[fill=white] at (3,0) {$\dots$};
      \node[tensor, red,label=above:$X$] (x) at (1.5,0) {};
    \end{tikzpicture} \ .
  \end{equation*}
  In particular,
  \begin{equation*} 
    \begin{tikzpicture}
      \draw (0.5,0) rectangle (4.5,-0.6);
      \foreach \x/\t in {1/1,2/2,4/n}{
        \node[tensor,label=below:$A_\t$] (t\x) at (\x,0) {};
        \draw[very thick] (t\x) --++ (0,0.5);
      }
      \node[fill=white] at (3,0) {$\dots$};
      \node[tensor, red,label=below:$Y$] (x) at (1.5,-0.6) {};
    \end{tikzpicture} = 
    \begin{tikzpicture}
      \draw (0.5,0) rectangle (4.5,-0.6);
      \foreach \x/\t in {1/1,2/2,4/n}{
        \node[tensor,label=below:$\tilde{B}_\t$] (t\x) at (\x,0) {};
        \draw[very thick] (t\x) --++ (0,0.5);
      }
      \node[fill=white] at (3,0) {$\dots$};
      \node[tensor, red,label=below:$Y$] (x) at (1.5,-0.6) {};
    \end{tikzpicture} \ .
  \end{equation*}
  Let us now block the MPS into a two partite MPS:
  \begin{equation*}
    \ket{\Psi} = 
    \begin{tikzpicture}
      \draw (0.5,0) rectangle (2.5,-0.5);
      \foreach \x/\t in {1/A_1,2/a}{
        \node[tensor,label=below:$\t$] (t\x) at (\x,0) {};
        \draw[very thick] (t\x) --++ (0,0.5);
      }
    \end{tikzpicture} = 
    \begin{tikzpicture}
      \draw (0.5,0) rectangle (2.5,-0.5);
      \foreach \x/\t in {1/B_1,2/b}{
        \node[tensor,label=below:$\t$] (t\x) at (\x,0) {};
        \draw[very thick] (t\x) --++ (0,0.5);
      }
    \end{tikzpicture}\ ,  
  \end{equation*}
  with 
  \begin{align*} 
    \begin{tikzpicture}
      \node[tensor,label=below:$a$] (t) at (1,0) {};
      \draw[very thick] (t) --++ (0,0.5);
      \draw (0.5,0)--(1.5,0);
    \end{tikzpicture} & = 
    \begin{tikzpicture}
      \draw (1.5,0) -- (4.5,0);
      \foreach \x/\t in {2/2,4/n}{
        \node[tensor,label=below:$A_\t$] (t\x) at (\x,0) {};
        \draw[very thick] (t\x) --++ (0,0.5);
      }
      \node[fill=white] at (3,0) {$\dots$};
    \end{tikzpicture} \\
    \begin{tikzpicture}
      \node[tensor,label=below:$b$] (t) at (1,0) {};
      \draw[very thick] (t) --++ (0,0.5);
      \draw (0.5,0)--(1.5,0);
    \end{tikzpicture} & = 
    \begin{tikzpicture}
      \draw (1.5,0) -- (4.5,0);
      \foreach \x/\t in {2/2,4/n}{
        \node[tensor,label=below:$\tilde{B}_\t$] (t\x) at (\x,0) {};
        \draw[very thick] (t\x) --++ (0,0.5);
      }
      \node[fill=white] at (3,0) {$\dots$};
    \end{tikzpicture} \ . 
  \end{align*}  
  After this blocking, the requirements of \cref{lem:inj_equal_tensors} are satisfied, therefore $A_1 = \lambda_1 \tilde{B}_1$. Similarly for all $i$,  $A_i = \lambda_i \tilde{B}_i$ and $\prod_i \lambda_i = 1$. Notice that these $\lambda_i$ can be sequentially absorbed into the gauges $Z_i$ in \cref{eq:B tilde}. 
\end{proof}

Notice that if the two MPS are translational invariant, i.e.\ the tensors at each vertex are the same, then the gauges relating them are also translational invariant (up to a constant), as 
\begin{equation*}
  \begin{tikzpicture}
    \draw (-1,0)--(1,0);
    \node[tensor,label=below:$Z_{i-1}^{-1}$] at (-0.5,0) {};
    \node[tensor,label=below:$Z_{i}$\vphantom{$Z_{i-1}^{-1}$}] at (0.5,0) {};
    \node[tensor,label=below:$A$\vphantom{$Z_{i-1}^{-1}$}] (t) at (0,0) {};
    \draw[very thick] (t)--(0,0.5);
  \end{tikzpicture} \ =
  \begin{tikzpicture}
    \draw (-1,0)--(1,0);
    \node[tensor,label=below:$Z_i^{-1}$] at (-0.5,0) {};
    \node[tensor,label=below:$Z_{i+1}$\vphantom{$Z_i^{-1}$}] at (0.5,0) {};
    \node[tensor,label=below:$A$\vphantom{$Z_i^{-1}$}] (t) at (0,0) {};
    \draw[very thick] (t)--(0,0.5);
  \end{tikzpicture} \ \Rightarrow \ Z_i \propto Z_{i+1},
\end{equation*}  
which can be seen by inverting the tensor $A$. We conclude therefore that
\begin{corollary}
  Let the tensors $A$ and $B$ define two injective, translational invariant MPS on $n\geq 3$ particles. Suppose they generate the same state:
  \begin{equation*}
    \ket{\Psi} = 
    \begin{tikzpicture}
      \draw (0.5,0) rectangle (4.5,-0.5);
      \foreach \x/\t in {1/1,2/2,4/n}{
        \node[tensor,label=below:$A$] (t\x) at (\x,0) {};
        \draw[very thick] (t\x) --++ (0,0.5);
      }
      \node[fill=white] at (3,0) {$\dots$};
    \end{tikzpicture} = 
    \begin{tikzpicture}
      \draw (0.5,0) rectangle (4.5,-0.5);
      \foreach \x/\t in {1/1,2/2,4/n}{
        \node[tensor,label=below:$B$] (t\x) at (\x,0) {};
        \draw[very thick] (t\x) --++ (0,0.5);
      }
      \node[fill=white] at (3,0) {$\dots$};
    \end{tikzpicture} \ .
  \end{equation*}  
  Then there is an invertible matrix $Z$ and a constant $\lambda\in\mathbb{C}$, $\lambda^n = 1$, such that 
    \begin{equation*}
      \begin{tikzpicture}
        \draw (-0.5,0)--(0.5,0);
        \node[tensor,label=below:$B$] (t) at (0,0) {};
        \draw[very thick] (t)--(0,0.5);
      \end{tikzpicture}  =  \lambda \cdot\ 
      \begin{tikzpicture}
        \draw (-1,0)--(1,0);
        \node[tensor,label=below:$Z^{-1}$] at (-0.5,0) {};
        \node[tensor,label=below:$Z$\vphantom{$Z^{-1}$}] at (0.5,0) {};
        \node[tensor,label=below:$A$\vphantom{$Z^{-1}$}] (t) at (0,0) {};
        \draw[very thick] (t)--(0,0.5);
      \end{tikzpicture} \ .
    \end{equation*}    
    Moreover, the gauge $Z$ is unique up to a multiplicative constant.    
\end{corollary}

\section{Injective PEPS}

In general, PEPS can be defined on any graph (no double edges are allowed, but there are extra edges attached to every vertex that is associated to a physical particle). The state corresponding to the PEPS is obtained by placing tensors on each vertex and contracting all indices corresponding to the edges of the graph. An example of a tensor network is depicted below:
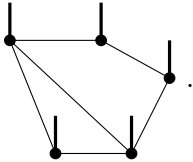
\begin{equation*}
  \begin{tikzpicture}
    \coordinate (c1) at (0,0);
    \coordinate (c2) at (1,0);
    \coordinate (c3) at (1.5,1);
    \coordinate (c4) at (0.6,1.5);
    \coordinate (c5) at (-0.6,1.5);
    \foreach \x in {1,2,3,4,5}{
      \filldraw (c\x) circle (0.07);
      \draw[very thick] (c\x) --++(0,0.5);
    }
    \draw (c1)--(c2)--(c3)--(c4)--(c5)--(c1);
    \draw (c2)--(c5);
  \end{tikzpicture} \ .
\end{equation*}
This definition includes TNs such as MPS, 2D PEPS and higher-dimensional PEPS. It also includes PEPS defined on arbitrary lattices, such as hyperbolic lattices used in the AdS/CFT correspondence\cite{Pastawski2015,Hayden2016}.  

We say that the tensor network is \emph{injective} if all tensors interpreted as maps from the virtual space to the physical one are injective. This is equivalent to the tensor having a one-sided inverse, as in the MPS case. Similar to the MPS case, the contraction of two injective tensors results in an injective tensor. 

One can group particles of the PEPS together treating them as one bigger particle. This regrouping can naturally be reflected in PEPS. In particular, we will block tensor networks to a three particle MPS as follows. Choose one edge of the PEPS and group together all vertices except the endpoints of the edge. This regrouped tensor together with the two endpoints of the edge forms a three-partite MPS as illustrated below; notice that the resulting MPS is injective: 
\begin{equation}\label{eq:block_to_mps}
  \begin{tikzpicture}
    \coordinate (c1) at (0,0);
    \coordinate (c2) at (1,0);
    \coordinate (c3) at (1.5,1);
    \coordinate (c4) at (0.6,1.5);
    \coordinate (c5) at (-0.6,1.5);
    \foreach \x in {1,2,3,4,5}{
      \filldraw (c\x) circle (0.07);
      \draw[very thick] (c\x) --++(0,0.5);
    }
    \draw (c1)--(c2)--(c3)--(c4)--(c5)--(c1);
    \draw (c2)--(c5);
    \draw[red] (c1) circle (0.3cm);
    \draw[red] (c5) circle (0.3cm);
    \draw[rounded corners=2mm,red] (0.4,-0.2) rectangle (1.7,1.7);
    \node[red] at (-0.45,-0.35) {$A'_2$};
    \node[red] at (-1.1,1.8) {$A'_1$};
    \node[red] at (1.7,1.8) {$A'_3$};
  \end{tikzpicture} \ \Rightarrow \ 
  \begin{tikzpicture}
    \draw (0.5,0) rectangle (3.5,-0.6);
    \foreach \x in {1,2,3}{
      \node[tensor,label=below:$A'_\x$] (t\x) at (\x,0) {};
      \draw[very thick] (t\x) --++ (0,0.5);
    }
  \end{tikzpicture}  \ .  
\end{equation}
Consider now two injective PEPS defined on the same graph that generate the same state:
\begin{equation}\label{eq:TN_5_particle_eq}
  \begin{tikzpicture}
    \node[tensor, label=below:$A_1$] (c1) at (0,0) {};
    \node[tensor, label=below:$A_2$] (c2) at (1,0) {};
    \node[tensor, label=below right:$A_3$] (c3) at (1.5,1) {}; 
    \node[tensor, label=below:$A_4$] (c4) at (0.6,1.5) {};
    \node[tensor, label=below left:$A_5$] (c5) at (-0.6,1.5) {};
    \foreach \x in {1,2,3,4,5}{
      \draw[very thick] (c\x) --++(0,0.5);
    }
    \draw (c1)--(c2)--(c3)--(c4)--(c5)--(c1);
    \draw (c2)--(c5);
  \end{tikzpicture} \ = \ 
  \begin{tikzpicture}
    \node[tensor, label=below:$B_1$] (c1) at (0,0) {};
    \node[tensor, label=below:$B_2$] (c2) at (1,0) {};
    \node[tensor, label=below right:$B_3$] (c3) at (1.5,1) {};
    \node[tensor, label=below:$B_4$] (c4) at (0.6,1.5) {};
    \node[tensor, label=below left:$B_5$] (c5) at (-0.6,1.5) {};
    \foreach \x in {1,2,3,4,5}{
      \draw[very thick] (c\x) --++(0,0.5);
    }
    \draw (c1)--(c2)--(c3)--(c4)--(c5)--(c1);
    \draw (c2)--(c5);
  \end{tikzpicture} \ .
\end{equation}
After blocking to MPS as described above, we arrive at two injective MPS generating the same state; hence \cref{lem:inj_isomorph} can be applied. This establishes a gauge transformation on the edge $(1,5)$ of the original PEPS. Similar regrouping can be done around every edge; applying then \cref{lem:inj_isomorph} results in a gauge transformation assigned to every edge. Define now the tensors $\tilde{B}_i$ by absorbing these gauges into the tensors $B_i$. For the resulting PEPS, we have that for every edge and matrix $X$  
\begin{equation}\label{eq:inj_equal_edge}
  \begin{tikzpicture}
    \node[tensor, label=below:$A_1$] (c1) at (0,0) {};
    \node[tensor, label=below:$A_2$] (c2) at (1,0) {};
    \node[tensor, label=below right:$A_3$] (c3) at (1.5,1) {};
    \node[tensor, label=below:$A_4$] (c4) at (0.6,1.5) {};
    \node[tensor, label=below left:$A_5$] (c5) at (-0.6,1.5) {};
    \foreach \x in {1,2,3,4,5}{
      \draw[very thick] (c\x) --++(0,0.5);
    }
    \draw (c1)--(c2)--(c3)--(c4)--(c5)--(c1);
    \draw (c2)--(c5);
    \node[red,tensor, label=right:$X$] at ($(c2)!0.5!(c5)$) {};
  \end{tikzpicture} \ = \ 
  \begin{tikzpicture}
    \node[tensor, label=below:$\tilde{B}_1$] (c1) at (0,0) {};
    \node[tensor, label=below:$\tilde{B}_2$] (c2) at (1,0) {};
    \node[tensor, label=below right:$\tilde{B}_3$] (c3) at (1.5,1) {};
    \node[tensor, label=below:$\tilde{B}_4$] (c4) at (0.6,1.5) {};
    \node[tensor, label=below left:$\tilde{B}_5$] (c5) at (-0.6,1.5) {};
    \foreach \x in {1,2,3,4,5}{
      \draw[very thick] (c\x) --++(0,0.5);
    }
    \draw (c1)--(c2)--(c3)--(c4)--(c5)--(c1);
    \draw (c2)--(c5);
    \node[red,tensor, label=right:$X$] at ($(c2)!0.5!(c5)$) {};
  \end{tikzpicture} \ .
\end{equation}
To conclude that $A_i = \lambda_i \tilde{B}_i$, we will need to use a more general version of \cref{lem:inj_equal_tensors}:
\begin{lemma}\label{lem:inj_equal_tensors_2}
  Let $A_1,A_2$ and $B_1,B_2$ be injective tensors. Suppose for all $X$ on all edges
  \begin{equation}\label{eq:lem_inj_eq_ten_2}
    \begin{tikzpicture}
      \foreach \x/\t in {1/1,3/2}{
        \node[tensor,label=below:$A_\t$] (t\t) at (\x,0) {};
        \draw[very thick] (t\x) --++ (0,0.5);
      }
      \draw (t1) to[bend left=30] node[midway,tensor,red,label=above:$X$] {} (t2);
      \draw (t1) to[bend right=30] (t2);
      \draw (t1) to[bend right=45] (t2);
      \node at (2,0.05) {$\vdots$};
    \end{tikzpicture} = 
    \begin{tikzpicture}
      \foreach \x/\t in {1/1,3/2}{
        \node[tensor,label=below:$B_\t$] (t\t) at (\x,0) {};
        \draw[very thick] (t\x) --++ (0,0.5);
      }
      \draw (t1) to[bend left=30] node[midway,tensor,red,label=above:$X$] {} (t2);
      \draw (t1) to[bend right=30] (t2);
      \draw (t1) to[bend right=45] (t2);
      \node at (2,0.05) {$\vdots$};
    \end{tikzpicture}  \ .    
  \end{equation}
  Then $A_1 = \lambda B_1$ and $A_2 = \lambda^{-1} B_2$ for some constant $\lambda$.  
\end{lemma}
\begin{proof}
  W.l.o.g.\ suppose that there are three lines connecting the tensors. Similar to the proof of \cref{lem:inj_equal_tensors}, if \cref{eq:lem_inj_eq_ten_2} holds for all $X$, then
  \begin{align*}
    \begin{tikzpicture}
      \foreach \x/\t in {1/1,3/2}{
        \node[tensor,label=below:$A_\t$] (t\t) at (\x,0) {};
        \draw[very thick] (t\x) --++ (0,0.5);
      }
      \draw (t1) to[bend left=30]  (t2);
      \draw (t1) to[bend right=10] (t2);
      \draw (t1) to[bend right=45] (t2);
      \fill[white] (1.3,0.1) rectangle (2.7,0.4);
    \end{tikzpicture} &= 
    \begin{tikzpicture}
      \foreach \x/\t in {1/1,3/2}{
        \node[tensor,label=below:$B_\t$] (t\t) at (\x,0) {};
        \draw[very thick] (t\x) --++ (0,0.5);
      }
      \draw (t1) to[bend left=30]  (t2);
      \draw (t1) to[bend right=10] (t2);
      \draw (t1) to[bend right=45] (t2);
      \fill[white] (1.3,0.1) rectangle (2.7,0.4);
    \end{tikzpicture}  \\
    \begin{tikzpicture}
      \foreach \x/\t in {1/1,3/2}{
        \node[tensor,label=below:$A_\t$] (t\t) at (\x,0) {};
        \draw[very thick] (t\x) --++ (0,0.5);
      }
      \draw (t1) to[bend left=30]  (t2);
      \draw (t1) to[bend right=10] (t2);
      \draw (t1) to[bend right=45] (t2);
      \fill[white] (1.3,-0.15) rectangle (2.7,-0.5);
    \end{tikzpicture} &= 
    \begin{tikzpicture}
      \foreach \x/\t in {1/1,3/2}{
        \node[tensor,label=below:$B_\t$] (t\t) at (\x,0) {};
        \draw[very thick] (t\x) --++ (0,0.5);
      }
      \draw (t1) to[bend left=30]  (t2);
      \draw (t1) to[bend right=10] (t2);
      \draw (t1) to[bend right=45] (t2);
      \fill[white] (1.3,-0.15) rectangle (2.7,-0.5);
    \end{tikzpicture}  \\ 
    \begin{tikzpicture}
      \foreach \x/\t in {1/1,3/2}{
        \node[tensor,label=below:$A_\t$] (t\t) at (\x,0) {};
        \draw[very thick] (t\x) --++ (0,0.5);
      }
      \draw (t1) to[bend left=30]  (t2);
      \draw (t1) to[bend right=10] (t2);
      \draw (t1) to[bend right=45] (t2);
      \fill[white] (1.3,-0.15) rectangle (2.7,0.1);
    \end{tikzpicture} &= 
    \begin{tikzpicture}
      \foreach \x/\t in {1/1,3/2}{
        \node[tensor,label=below:$B_\t$] (t\t) at (\x,0) {};
        \draw[very thick] (t\x) --++ (0,0.5);
      }
      \draw (t1) to[bend left=30]  (t2);
      \draw (t1) to[bend right=10] (t2);
      \draw (t1) to[bend right=45] (t2);
      \fill[white] (1.3,-0.15) rectangle (2.7,0.1);
    \end{tikzpicture}  \ .              
  \end{align*} 
  Applying now the inverse of $A_2$, we conclude that 
  \begin{equation*}
    \begin{tikzpicture}
      \node[tensor, label=below:$A_1$] (a) {};
      \draw[very thick] (a) --++ (0,0.5);
      \draw (a) -- (30:0.8);
      \draw (a) -- (-10:0.8);
      \draw (a) -- (-45:0.8);
    \end{tikzpicture} = 
    \begin{tikzpicture}
      \node[tensor, label=below:$B_1$] (a) {};
      \draw[very thick] (a) --++ (0,0.5);
      \draw (a) -- (30:0.8);
      \draw (a) -- (-10:0.8);
      \draw (a) -- (-45:0.8);
      \draw[red, line width=1mm] (-55:0.6) arc (-55:0:0.6);
      \node at (-27.5:0.9) {$Z$};
    \end{tikzpicture} = 
    \begin{tikzpicture}
      \node[tensor, label=below:$B_1$] (a) {};
      \draw[very thick] (a) --++ (0,0.5);
      \draw (a) -- (30:0.8);
      \draw (a) -- (-10:0.8);
      \draw (a) -- (-45:0.8);
      \draw[red, line width=1mm] (-20:0.6) arc (-20:40:0.6);
      \node at (10:0.9) {$U$};
    \end{tikzpicture} = 
    \begin{tikzpicture}
      \node[tensor, label=below:$B_1$] (a) {};
      \draw[very thick] (a) --++ (0,0.5);
      \draw (a) -- (30:0.8);
      \draw (a) -- (-10:0.4);
      \draw (a) -- (-45:0.8);
      \draw[red, line width=1mm] (-55:0.6) arc (-55:40:0.6);
      \node at (-7.5:0.9) {$W$};
    \end{tikzpicture} \ . 
  \end{equation*}
  Inverting $B_1$ we conclude that the gauges $Z,U,W$ satisfy 
  \begin{equation*}
    \sum_i \id \otimes Z^{(1)}_i \otimes Z^{(2)}_i= \sum_i  U^{(1)}_i \otimes U^{(2)}_i \otimes \id = \sum_i  W^{(1)}_i \otimes \id  \otimes W^{(2)}_i , 
  \end{equation*}
  where we have written 
  \begin{align*}
    Z &=\sum_i Z_i^{(1)}\otimes Z_i^{(2)}\\
    U &=\sum_i U_i^{(1)}\otimes U_i^{(2)}\\ 
    W &=\sum_i W_i^{(1)}\otimes W_i^{(2)} \ .
  \end{align*}
  Therefore all three gauges are proportional to the identity and thus $A_1 = \lambda B_1$. Similarly we get $A_2 = 1/\lambda B_2$.
\end{proof}

Let us now block the PEPS in \cref{eq:inj_equal_edge} into two injective tensors: select one tensor and block all the others into another injective tensor. These PEPS now satisfy the requirements of  \cref{lem:inj_equal_tensors_2} and thus for all $i$, $A_i = \lambda_i \tilde{B}_i$ for some constant $\lambda_i$, giving the Fundamental Theorem for general injective PEPS (the constants $\lambda_i$ can be incorporated into the gauge transformations):
\begin{theorem}\label{thm:inj}
  Two injective PEPS  -- defined on a graph that does not contain double edges and self-loops -- generate the same state if and only if the generating tensors are related with a local gauge. These gauges are unique up to a multiplicative constant. 
\end{theorem}

As the defining graph can not contain double edges and self-loops, the theorem is applicable for MPS of size $N$ only if $N\geq 3$, and for 2D PEPS of size $N\times M$ only if both $N\geq 3$ and $M\geq 3$. As an illustration of the theorem, for the two PEPS in \cref{eq:TN_5_particle_eq} there are matrices $Z_{12},Z_{23},Z_{34},Z_{45},Z_{51}$ and $Z_{25}$ such that
\begin{align*}
  \begin{tikzpicture}
    \node[tensor, label=below:$A_1$] (c1) at (0,0) {};
    \coordinate (c2) at (1,0) {};
    \coordinate (c5) at (-0.6,1.5) {};
    \draw[very thick] (c1) --++(0,0.5);
    \draw (c1)--(c2);
    \draw ($(c1)!0.5!(c5)$)--(c1);
    \node[red,tensor,label=below:$Z_{12}$] (z12) at ($(c1)!0.5!(c2)$) {}; 
    \node[red,tensor,label=below left:$Z_{51}$] (z15) at ($(c1)!0.25!(c5)$) {}; 
  \end{tikzpicture}    =
  \begin{tikzpicture}
    \node[tensor, label=below:$B_1$] (c1) at (0,0) {};
    \coordinate (c2) at (1,0) {};
    \coordinate (c5) at (-0.6,1.5) {};
    \draw[very thick] (c1) --++(0,0.5);
    \draw (c1)--(c2);
    \draw ($(c1)!0.5!(c5)$)--(c1);
  \end{tikzpicture}  \quad &\text{and} \quad
  \begin{tikzpicture}
    \coordinate (c1) at (0,0);
    \node[tensor, label=below:$A_2$] (c2) at (1,0) {};
    \coordinate (c3) at (1.5,1);
    \coordinate (c5) at (-0.6,1.5) {};
    \draw[very thick] (c2) --++(0,0.5);
    \draw (c1)--(c2);
    \draw (c2)--(c3);
    \draw (c2)--($(c2)!0.5!(c5)$);
    \node[red,tensor,label=below:$Z_{12}^{-1}$] (z12) at ($(c1)!0.5!(c2)$) {}; 
    \node[red,tensor,label=right:$Z_{23}$] (z23) at ($(c3)!0.5!(c2)$) {}; 
    \node[red,tensor,label=left:$Z_{25}$] (z25) at ($(c2)!0.25!(c5)$) {}; 
  \end{tikzpicture}   = 
  \begin{tikzpicture}
    \coordinate (c1) at (0,0);
    \node[tensor, label=below:$B_2$] (c2) at (1,0) {};
    \coordinate (c3) at (1.5,1);
    \coordinate (c5) at (-0.6,1.5) {};
    \draw[very thick] (c2) --++(0,0.5);
    \draw (c1)--(c2);
    \draw (c2)--(c3);
    \draw (c2)--($(c2)!0.5!(c5)$);
  \end{tikzpicture}   \ , \\
  \begin{tikzpicture}
    \coordinate (c2) at (1,0);
    \node[tensor, label=right:$A_3$] (c3) at (1.5,1) {};
    \coordinate (c4) at (0.6,1.5);
    \draw[very thick] (c3) --++(0,0.5);
    \draw (c2)--(c3);
    \draw (c3)--(c4);
    \node[red,tensor,label=right:$Z_{23}^{-1}$] (z23) at ($(c3)!0.5!(c2)$) {}; 
    \node[red,tensor,label=left:$Z_{34}$] (z34) at ($(c3)!0.5!(c4)$) {}; 
  \end{tikzpicture}   = 
  \begin{tikzpicture}
    \coordinate (c2) at (1,0);
    \node[tensor, label=right:$B_3$] (c3) at (1.5,1) {};
    \coordinate (c4) at (0.6,1.5);
    \draw[very thick] (c3) --++(0,0.5);
    \draw (c2)--(c3);
    \draw (c3)--(c4);
  \end{tikzpicture}  \quad &\text{and} \quad
  \begin{tikzpicture}
    \coordinate (c3) at (1.5,1);
    \node[tensor, label=below:$A_4$] (c4) at (0.6,1.5) {};
    \coordinate (c5) at (-0.6,1.5);
    \draw[very thick] (c4) --++(0,0.5);
    \draw (c3)--(c4);
    \draw (c4)--(c5);
    \node[red,tensor,label=right:$Z_{34}^{-1}$] (z34) at ($(c3)!0.5!(c4)$) {}; 
    \node[red,tensor,label=below:$Z_{45}$] (z45) at ($(c4)!0.5!(c5)$) {}; 
  \end{tikzpicture}  = 
  \begin{tikzpicture}
    \coordinate (c3) at (1.5,1);
    \node[tensor, label=below:$B_4$] (c4) at (0.6,1.5) {};
    \coordinate (c5) at (-0.6,1.5);
    \draw[very thick] (c4) --++(0,0.5);
    \draw (c3)--(c4);
    \draw (c4)--(c5);
  \end{tikzpicture} \ ,  \\
  \begin{tikzpicture}
    \coordinate (c1) at (0,0);
    \coordinate (c2) at (1,0);
    \coordinate (c4) at (0.6,1.5);
    \node[tensor, label=left:$A_5$] (c5) at (-0.6,1.5) {};
    \draw[very thick] (c5) --++(0,0.5);
    \draw (c4)--(c5);
    \draw (c5)--($(c1)!0.5!(c5)$);
    \draw ($(c2)!0.5!(c5)$)--(c5);
    \node[red,tensor,label=left:$Z_{51}^{-1}$] (z51) at ($(c5)!0.25!(c1)$) {}; 
    \node[red,tensor,label=right:$Z_{25}^{-1}$] (z25) at ($(c5)!0.25!(c2)$) {}; 
    \node[red,tensor,label=above:$Z_{45}^{-1}$] (z45) at ($(c5)!0.5!(c4)$) {}; 
  \end{tikzpicture} & = 
  \begin{tikzpicture}
    \coordinate (c1) at (0,0);
    \coordinate (c2) at (1,0);
    \coordinate (c4) at (0.6,1.5);
    \node[tensor, label=left:$B_5$] (c5) at (-0.6,1.5) {};
    \draw[very thick] (c5) --++(0,0.5);
    \draw (c4)--(c5);
    \draw (c5)--($(c1)!0.5!(c5)$);
    \draw ($(c2)!0.5!(c5)$)--(c5);
  \end{tikzpicture} \ .
\end{align*} 

\section{Normal PEPS}\label{sec:normal}

We call a PEPS \emph{normal}, if blocking tensors in certain regions results in injective tensors. To derive the Fundamental Theorem for this kind of PEPS, we use the same arguments as above after blocking tensors to injective ones. This technique requires that the system is big enough to allow for blocking. This proof technique presented here is not optimal in the required system size; we describe a proof technique giving tighter bounds in \cref{sec:normal_alt}. For simplicity, we present the proof for a TI normal PEPS on a square lattice, but it can easily be generalized to the non TI case on any geometry.

Before proceeding to the proof, we need the following lemma:

\begin{lemma}\label{lem:injective_union}
  The union of two injective regions is injective.
\end{lemma}

\begin{proof}
  Let $A$ and $B$ be two injective regions. W.l.o.g.\ the TN can be blocked as follows (missing edges don't change the proof):
  
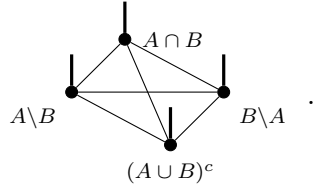
\begin{equation*}
    \begin{tikzpicture}
      \node[tensor,label=below left:$A\backslash B$] (1) at (0,0) {};
      \node[tensor,label=right:$A\cap B$] (2) at (0.7,0.7) {};
      \node[tensor,label=below right:$B\backslash A$] (3) at (2,0) {};
      \node[tensor,label=below:$(A\cup B)^{c}$] (4) at (1.3,-0.7) {};
      \draw[very thick] (1)--++(0,0.5);
      \draw[very thick] (2)--++(0,0.5);
      \draw[very thick] (3)--++(0,0.5);
      \draw[very thick] (4)--++(0,0.5);
      \draw (1)--(2)--(3)--(4)--(1)--(3);
      \draw (2) -- (4);
    \end{tikzpicture} \ .
  \end{equation*}
  Notice that $A\cup B = (A\backslash B) \cup (A\cap B) \cup (B\backslash A)$. Let $X$ now be a tensor such that 
  
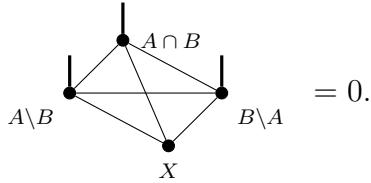
\begin{equation*}
    \begin{tikzpicture}
      \node[tensor,label=below left:$A\backslash B$] (1) at (0,0) {};
      \node[tensor,label=right:$A\cap B$] (2) at (0.7,0.7) {};
      \node[tensor,label=below right:$B\backslash A$] (3) at (2,0) {};
      \node[tensor,label=below:$X$] (4) at (1.3,-0.7) {};
      \draw[very thick] (1)--++(0,0.5);
      \draw[very thick] (2)--++(0,0.5);
      \draw[very thick] (3)--++(0,0.5);
      \draw (1)--(2)--(3)--(4)--(1)--(3);
      \draw (2) -- (4);
    \end{tikzpicture} \ = 0.
  \end{equation*}
  As the region $A = (A\backslash B) \cup (A\cap B)$ is injective,
  
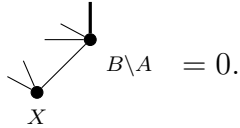
\begin{equation*}
    \begin{tikzpicture}
      \coordinate (1) at (0,0);
      \coordinate (2) at (0.7,0.7) {};
      \node[tensor,label=below right:$B\backslash A$] (3) at (2,0) {};
      \node[tensor,label=below:$X$] (4) at (1.3,-0.7) {};
      \draw[very thick] (3)--++(0,0.5);
      \draw (3)--(4);
      \draw (4) -- ($(4)!0.3!(2)$);
      \draw (4) -- ($(4)!0.3!(1)$);
      \draw (3) -- ($(3)!0.3!(2)$);
      \draw (3) -- ($(3)!0.3!(1)$);
    \end{tikzpicture} \ = 0.
  \end{equation*}
  Plugging back the tensor over the region $A\cap B$,
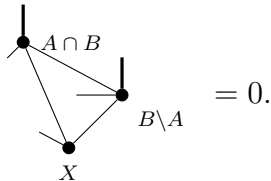
  \begin{equation*}
    \begin{tikzpicture}
      \coordinate (1) at (0,0);
      \node[tensor,label=right:$A\cap B$] (2) at (0.7,0.7) {};
      \node[tensor,label=below right:$B\backslash A$] (3) at (2,0) {};
      \node[tensor,label=below:$X$] (4) at (1.3,-0.7) {};
      \draw[very thick] (2)--++(0,0.5);
      \draw[very thick] (3)--++(0,0.5);
      \draw (3)--(4);
      \draw (4) -- (2);
      \draw (4) -- ($(4)!0.3!(1)$);
      \draw (3) -- (2);
      \draw (3) -- ($(3)!0.3!(1)$);
      \draw (2) -- ($(2)!0.3!(1)$);
    \end{tikzpicture} \ = 0.
  \end{equation*}
  Finally, the region $B = (A\cap B) \cup (B\backslash A)$ is injective, hence inverting the tensor over that region gives
  
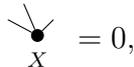
\begin{equation*}
    \begin{tikzpicture}
      \coordinate (1) at (0,0);
      \coordinate (2) at (0.7,0.7);
      \coordinate (3) at (2,0);
      \node[tensor,label=below:$X$] (4) at (1.3,-0.7) {};
      \draw (4) -- ($(4)!0.3!(1)$);
      \draw (4) -- ($(4)!0.3!(2)$);
      \draw (4) -- ($(4)!0.3!(3)$);
    \end{tikzpicture} \ = 0,
  \end{equation*}
  which means that the region $A\cup B$ is injective. 
\end{proof}

For example, if in a TI 2D PEPS every $2\times 3$ and $3\times 2$ region is injective, then the following regions:
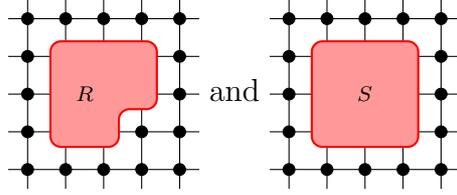
\begin{equation*}
  \begin{tikzpicture}
    \clip (-0.25,0.75) rectangle (2.25,3.25);
    \draw[step=0.5] (-1,-1) grid (4,4);
    \foreach \x in {0,0.5,...,3}{
      \foreach \y in {0,0.5,...,3}{
        \node[tensor] at (\x,\y) {};
      }
    }
    \draw[red,fill=red!40, thick, rounded corners] (0.3,1.3) -- (1.2,1.3) --(1.2,1.8)--(1.7,1.8)--(1.7,2.7)--(0.3,2.7)--cycle;
    \node  at (0.75,2) {$R$}; 
  \end{tikzpicture}  \ {\rm and} \     
  \begin{tikzpicture}
    \clip (-0.25,0.75) rectangle (2.25,3.25);
    \draw[step=0.5] (-1,-1) grid (4,4);
    \foreach \x in {0,0.5,...,3}{
      \foreach \y in {0,0.5,...,3}{
        \node[tensor] at (\x,\y) {};
      }
    }
    \draw[red,fill=red!40, thick, rounded corners] (0.3,1.3) rectangle (1.7,2.7);
    \node  at (1,2) {$S$}; 
  \end{tikzpicture}  
\end{equation*}
are unions of smaller injective regions, and they are thus injective. Similarly, if the size of the PEPS is at least $5\times 6$, then the region $T$ depicted below is injective:
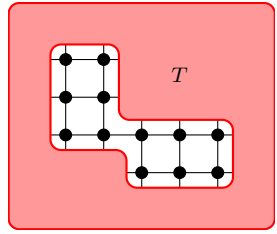
\begin{equation*}
  \begin{tikzpicture}
    \clip (-0.25,0.25) rectangle (3.25,3.25);
    \draw[step=0.5] (-1,-1) grid (4,4);
    \foreach \x in {0,0.5,...,3}{
      \foreach \y in {0,0.5,...,3}{
        \node[tensor] at (\x,\y) {};
      }
    }
    \filldraw[draw=red, fill=red!40, thick, rounded corners,even odd rule] (0.3,2.7) -- (1.2,2.7) -- (1.2,1.7) -- (2.7,1.7)--(2.7,0.8)--(1.3,0.8)--(1.3,1.3)--(0.3,1.3)--cycle
    (-0.25,0.25) rectangle (3.25,3.25)
    ;
    \node  at (2,2.3) {$T$}; 
  \end{tikzpicture} \ .
\end{equation*}

In the following we prove the Fundamental Theorem for a normal TI 2D PEPS. In particular, we prove it in detail for the case where every region of size $2\times 3$ and $3\times 2$ is injective as in the examples above. Then, we generalize the proof for any normal PEPS that is big enough to allow the necessary blockings.

\begin{theorem}\label{thm:3}
  Let $A$ and $B$ be two normal 2D PEPS tensors such that every $2\times 3$ and $3\times 2$ region is injective. Suppose they generate the same state on some region $n\times m$ with $n,m\geq 7$. Then $A$ and $B$ are related to each other with a gauge transformation: 
  \begin{equation*}
    \begin{tikzpicture}[baseline=-0.1cm]
      \draw (-0.5,0)--(0.5,0);
      \draw (0,0,-0.5)--(0,0,0.5);
      \node[tensor,label=below:$B$] at (0,0) {};
      \draw[very thick] (0,0)--++(0,0.5);
    \end{tikzpicture} = \lambda \cdot 
    \begin{tikzpicture}[baseline=-0.1cm]
      \draw (-1.3,0)--(1.3,0);
      \draw (0,0,-1.6)--(0,0,1.6);
      \node[tensor,label=below:$A$] at (0,0) {};
      \node[tensor,label=right:$Y^{-1}$] at (0,0,-1) {};
      \node[tensor,label=below:$Y$] at (0,0,1) {};
      \node[tensor,label=below:$X^{-1}$] at (0.8,0,0) {};
      \node[tensor,label=below:$X$] at (-0.8,0,0) {};
      \draw[very thick] (0,0)--++(0,0.5);
    \end{tikzpicture} \ , 
  \end{equation*}
  with $\lambda^{n\cdot m} = 1$ and $X,Y$ invertible matrices. $X$ and $Y$ are unique up to a multiplicative constant.
\end{theorem}

\begin{proof}
  Let us block the TN into three injective parts around an edge. This can be done with e.g. the following choice of regions:
  
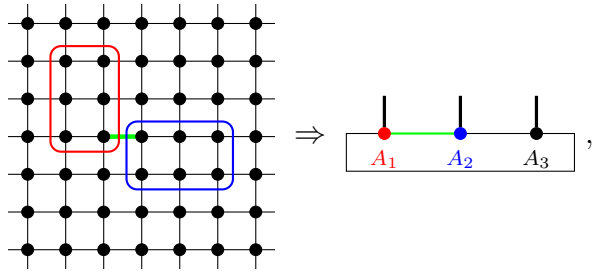
\begin{equation*}
    \begin{tikzpicture}
      \clip (-0.25,-0.25) rectangle (3.25,3.25);
      \draw[step=0.5] (-1,-1) grid (4,4);
      \draw[line width=0.6mm,green] (1,1.5)--(1.5,1.5);
      \foreach \x in {0,0.5,...,3}{
        \foreach \y in {0,0.5,...,3}{
          \node[tensor] at (\x,\y) {};
        }
      }
      \draw[red,thick, rounded corners] (0.3,1.3) rectangle (1.2,2.7);
      \draw[blue,thick, rounded corners] (1.3,1.7) rectangle (2.7,0.8);
    \end{tikzpicture} \ \Rightarrow \ 
    \begin{tikzpicture}
      \draw (0.5,0) rectangle (3.5,-0.5);
      \draw[thick,green] (1,0)--(2,0);
      \foreach \x/\c in {1/red,2/blue,3/black}{
        \node[tensor,\c,label=below:$\color{\c}A_\x$] (t\x) at (\x,0) {};
        \draw[very thick] (t\x) --++ (0,0.5);
      }
    \end{tikzpicture}  \ ,    
  \end{equation*} 
  where $A_1$ corresponds to the red region, $A_2$ to the blue one and $A_3$ to the rest. The region $A_3$ is injective as long as the size of the PEPS is at least $5\times 7$. Therefore a $7\times 7$ PEPS can be blocked to injective three partite MPS around every edge (including the vertical edges that require a PEPS size at least $7 \times 5$). Therefore \cref{lem:inj_isomorph} can be applied giving a gauge transformation on every edge.  Due to translation invariance, these gauges are described by the same matrix $X$ ($Y$) on all  horizontal (vertical) edges. 
  
  Define now $\tilde{B}$ by incorporating the local gauges into the tensors $B$, such as in the injective case:
  \begin{equation*}
    \begin{tikzpicture}[baseline=-0.1cm]
      \draw (-0.5,0)--(0.5,0);
      \draw (0,0,-0.5)--(0,0,0.5);
      \node[tensor,label=below:$\tilde{B}$] at (0,0) {};
      \draw[very thick] (0,0)--++(0,0.5);
    \end{tikzpicture} = 
    \begin{tikzpicture}[baseline=-0.1cm]
      \draw (-1.3,0)--(1.3,0);
      \draw (0,0,-1.6)--(0,0,1.6);
      \node[tensor,label=below:$B$] at (0,0) {};
      \node[tensor,label=right:$Y$] at (0,0,-1) {};
      \node[tensor,label=below:$Y^{-1}$] at (0,0,1) {};
      \node[tensor,label=below:$X$] at (0.8,0,0) {};
      \node[tensor,label=below:$X^{-1}$] at (-0.8,0,0) {};
      \draw[very thick] (0,0)--++(0,0.5);
    \end{tikzpicture} \ . 
  \end{equation*}
  The two PEPS tensors $A$ and $\tilde{B}$ generate the same state. Moreover, inserting a matrix $Z$ on any bond of the first PEPS gives the same state as inserting the same matrix $Z$ on the corresponding bond of the second PEPS. Remember that \cref{lem:injective_union} implies that both 
  
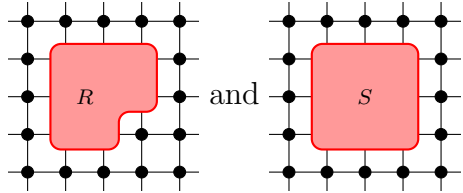
\begin{equation*}
    \begin{tikzpicture}
      \clip (-0.25,0.75) rectangle (2.25,3.25);
      \draw[step=0.5] (-1,-1) grid (4,4);
      \foreach \x in {0,0.5,...,3}{
        \foreach \y in {0,0.5,...,3}{
          \node[tensor] at (\x,\y) {};
        }
      }
      \draw[red,fill=red!40, thick, rounded corners] (0.3,1.3) -- (1.2,1.3) --(1.2,1.8)--(1.7,1.8)--(1.7,2.7)--(0.3,2.7)--cycle;
      \node  at (0.75,2) {$R$}; 
    \end{tikzpicture}  \ {\rm and} \     
    \begin{tikzpicture}
      \clip (-0.25,0.75) rectangle (2.25,3.25);
      \draw[step=0.5] (-1,-1) grid (4,4);
      \foreach \x in {0,0.5,...,3}{
        \foreach \y in {0,0.5,...,3}{
          \node[tensor] at (\x,\y) {};
        }
      }
      \draw[red,fill=red!40, thick, rounded corners] (0.3,1.3) rectangle (1.7,2.7);
      \node  at (1,2) {$S$}; 
    \end{tikzpicture}  
  \end{equation*}
  are injective regions and notice that the two regions differ in a single site. Moreover, if the PEPS is at least $5\times 5$, their complement regions $R^c$ and $S^c$ are also injective. Let us denote the tensor on region $R$ as $A_R$ ($\tilde{B}_R$) and on region $S$ as $A_S$ ($\tilde{B}_S$). Then, by \cref{lem:inj_equal_tensors_2}, $A_R \propto \tilde{B}_R$ and $A_S \propto \tilde{B}_S$. This can be represented as
  \begin{equation*}
    \begin{tikzpicture}
      \draw (0.5,0)--(2.5,0);
      \node[tensor,label=below:$A_{R}$] (a1) at (1,0) {};
      \draw[very thick] (a1)--++(0,0.5);
      \node[tensor,label=below:$A$] (a2) at (2,0) {};
      \draw[very thick] (a2)--++(0,0.5);
    \end{tikzpicture} = 
    \begin{tikzpicture}
      \draw (0.5,0)--(1.5,0);
      \node[tensor,label=below:$A_{S}$] (a1) at (1,0) {};
      \draw[very thick] (a1)--++(0,0.5);
    \end{tikzpicture} \propto
    \begin{tikzpicture}
      \draw (0.5,0)--(1.5,0);
      \node[tensor,label=below:$\tilde{B}_{S}$] (a1) at (1,0) {};
      \draw[very thick] (a1)--++(0,0.5);
    \end{tikzpicture} = 
    \begin{tikzpicture}
      \draw (0.5,0)--(2.5,0);
      \node[tensor,label=below:$\tilde{B}_{R}$] (a1) at (1,0) {};
      \draw[very thick] (a1)--++(0,0.5);
      \node[tensor,label=below:$\tilde{B}$] (a2) at (2,0) {};
      \draw[very thick] (a2)--++(0,0.5);
    \end{tikzpicture}.
  \end{equation*}
  Applying the inverse of $A_R\propto \tilde{B}_R$ on the two ends of the equation, we get that the tensors $A$ and $\tilde{B}$ are proportional.  
\end{proof}
 
The above proof can be repeated for any PEPS as long as it is possible to block into injective regions as required by \cref{lem:inj_isomorph} and \cref{lem:inj_equal_tensors_2}. This leads to the Fundamental Theorem of normal PEPS:

\begin{theorem}\label{thm:normal}
  Suppose two normal PEPS generating the same state satisfy the following:
  \begin{itemize}
  \item they can be blocked into three partite injective MPS around every edge,
  \item and for every site, there are injective regions with their complements also being injective that differ only in the given site.
  \end{itemize}
  Then the defining tensors are related with a local gauge. Moreover, the gauges are unique up to a multiplicative constant.
\end{theorem}
Notice that this statement holds for a fixed system size (which is big enough to allow blocking into injective MPS), and translational invariance is not required. In case of a translational invariant system, the gauges are also translational invariant (if the proportionality constants are not absorbed into the gauges). In the following we present some special cases. For non TI MPS, the statement reads as
\begin{corollary}
  Let $\{A_i\}_{i=1}^n$ and $\{B_i\}_{i=1}^n$ two normal MPS on $n\geq 3L$ sites with the property that blocking any $L$ consecutive sites results in an injective tensor. Suppose they generate the same state:
  \begin{equation*}
    \ket{\Psi} = 
    \begin{tikzpicture}
      \draw (0.5,0) rectangle (4.5,-0.5);
      \foreach \x/\t in {1/1,2/2,4/n}{
        \node[tensor,label=below:$A_\t$] (t\x) at (\x,0) {};
        \draw[very thick] (t\x) --++ (0,0.5);
      }
      \node[fill=white] at (3,0) {$\dots$};
    \end{tikzpicture} = 
    \begin{tikzpicture}
      \draw (0.5,0) rectangle (4.5,-0.5);
      \foreach \x/\t in {1/1,2/2,4/n}{
        \node[tensor,label=below:$B_\t$] (t\x) at (\x,0) {};
        \draw[very thick] (t\x) --++ (0,0.5);
      }
      \node[fill=white] at (3,0) {$\dots$};
    \end{tikzpicture} \ .
  \end{equation*}  
  Then there are invertible matrices $Z_i$ (for $i=1 \dots n$, $n+1\equiv 1$) such that for all $i=1\dots n$
    \begin{equation*}
      \begin{tikzpicture}
        \draw (-0.5,0)--(0.5,0);
        \node[tensor,label=below:$B_i$] (t) at (0,0) {};
        \draw[very thick] (t)--(0,0.5);
      \end{tikzpicture}  = 
      \begin{tikzpicture}
        \draw (-1,0)--(1,0);
        \node[tensor,label=below:$Z_i^{-1}$] at (-0.5,0) {};
        \node[tensor,label=below:$\ Z_{i+1}$\vphantom{$Z_i^{-1}$}] at (0.5,0) {};
        \node[tensor,label=below:$A_i$\vphantom{$Z_i^{-1}$}] (t) at (0,0) {};
        \draw[very thick] (t)--(0,0.5);
      \end{tikzpicture} \ .
    \end{equation*}    
    Moreover, the gauges $Z_i$ are unique up to a multiplicative constant.
\end{corollary}
In \cref{sec:normal_alt} we strengthen the statement to include system sizes $n\geq 2L +1$. For TI MPS, the statement reads as
\begin{corollary}
  Let $A$ and $B$ be two normal TI MPS on $n\geq 3L$ sites with the property that blocking $L$ consecutive sites results in an injective tensor. Suppose they generate the same state:
  \begin{equation*}
    \ket{\Psi} = 
    \begin{tikzpicture}
      \draw (0.5,0) rectangle (4.5,-0.5);
      \foreach \x/\t in {1/1,2/2,4/n}{
        \node[tensor,label=below:$A$] (t\x) at (\x,0) {};
        \draw[very thick] (t\x) --++ (0,0.5);
      }
      \node[fill=white] at (3,0) {$\dots$};
    \end{tikzpicture} = 
    \begin{tikzpicture}
      \draw (0.5,0) rectangle (4.5,-0.5);
      \foreach \x/\t in {1/1,2/2,4/n}{
        \node[tensor,label=below:$B$] (t\x) at (\x,0) {};
        \draw[very thick] (t\x) --++ (0,0.5);
      }
      \node[fill=white] at (3,0) {$\dots$};
    \end{tikzpicture} \ .
  \end{equation*}  
  Then there is an invertible matrix $Z$ and a constant $\lambda$ with $\lambda^n=1$ such that 
    \begin{equation*}
      \begin{tikzpicture}
        \draw (-0.5,0)--(0.5,0);
        \node[tensor,label=below:$B$] (t) at (0,0) {};
        \draw[very thick] (t)--(0,0.5);
      \end{tikzpicture}  = \lambda \cdot
      \begin{tikzpicture}
        \draw (-1,0)--(1,0);
        \node[tensor,label=below:$Z^{-1}$] at (-0.5,0) {};
        \node[tensor,label=below:$Z$\vphantom{$Z^{-1}$}] at (0.5,0) {};
        \node[tensor,label=below:$A$\vphantom{$Z^{-1}$}] (t) at (0,0) {};
        \draw[very thick] (t)--(0,0.5);
      \end{tikzpicture} \ .
    \end{equation*}    
    Moreover the gauge $Z$ is unique up to a multiplicative constant.
\end{corollary}
In \cref{sec:normal_alt} we strengthen the statement to include system sizes $n\geq 2L +1$. For 2D TI PEPS, the statement reads as 
\begin{corollary}
  Let $A$ and $B$ be two normal 2D PEPS tensors such that every $L \times K$ region is injective. Suppose they generate the same state on some region $n\times m$ with $n \geq 3 L$ and $m\geq 3K$. Then $A$ and $B$ are related to each other with a gauge: 
  \begin{equation*}
    \begin{tikzpicture}[baseline=-0.1cm]
      \draw (-0.5,0)--(0.5,0);
      \draw (0,0,-0.5)--(0,0,0.5);
      \node[tensor,label=below:$B$] at (0,0) {};
      \draw[very thick] (0,0)--++(0,0.5);
    \end{tikzpicture} = \lambda \cdot 
    \begin{tikzpicture}[baseline=-0.1cm]
      \draw (-1.3,0)--(1.3,0);
      \draw (0,0,-1.6)--(0,0,1.6);
      \node[tensor,label=below:$A$] at (0,0) {};
      \node[tensor,label=right:$Y^{-1}$] at (0,0,-1) {};
      \node[tensor,label=below:$Y$] at (0,0,1) {};
      \node[tensor,label=below:$X^{-1}$] at (0.8,0,0) {};
      \node[tensor,label=below:$X$] at (-0.8,0,0) {};
      \draw[very thick] (0,0)--++(0,0.5);
    \end{tikzpicture} \ , 
  \end{equation*}
  with $\lambda^{n\cdot m} = 1$ and $X,Y$ invertible matrices. Moreover these matrices $X,Y$ are unique up to a multiplicative constant.
\end{corollary}
In \cref{sec:normal_alt} we strengthen the statement to include system sizes $n\geq 2L +1$ and $m\geq 2K +1$. Similar statements can be made for the non-TI case as well as for other situations, including PEPS in 3 and higher dimensions, other lattices (e.g.\ triangular, honeycomb, Kagome), and other geometries (e.g.\ hyperbolic, as it is used in the AdS/CFT constructions \cite{Hayden2016,Pastawski2015}). 

Furthermore, the results hold for general tensor networks as well (including
tensors that do not have physical index), provided that the TN satisfies the conditions
in \cref{thm:normal}. However, there is an important class of TN that do not satisfy them, namely the MERA \cite{Vidal2007}, and thus our results do not apply to them.

\section{Applications}

In this Section we show how the above results can be applied in different scenarios. In particular, we consider local (gauge) and global symmetries as well as translation symmetry.

Consider a normal TN on $n$ particles describing a state $\ket{\Psi}$. Suppose $\ket{\Psi}$ admits a global symmetry: $U^{\otimes n} \Psi = \Psi$. Then, if the TN satisfies the conditions in \cref{thm:normal}, the symmetry operators acting on the individual tensors is the same as acting with gauge transformations on the virtual level. For example, in TI MPS, this is reflected as:
\begin{equation*}
  \begin{tikzpicture}[baseline=-0.1cm]
    \draw (-0.5,0)--(0.5,0);
    \node[tensor,label=below:$A$] (t) at (0,0) {};
    \draw[very thick] (t)--(0,1);
    \node[tensor,label=left:$U$] (t) at (0,0.5) {};
  \end{tikzpicture}  =  \lambda \cdot\ 
  \begin{tikzpicture}[baseline=-0.1cm]
    \draw (-1,0)--(1,0);
    \node[tensor,label=below:$Z^{-1}$] at (-0.5,0) {};
    \node[tensor,label=below:$Z$\vphantom{$Z^{-1}$}] at (0.5,0) {};
    \node[tensor,label=below:$A$\vphantom{$Z^{-1}$}] (t) at (0,0) {};
    \draw[very thick] (t)--(0,0.5);
  \end{tikzpicture} \ ,
\end{equation*}    
with $\lambda^n = 1$. Similar statements are true in the non TI case (in which case the gauges might be different on every edge) and for any geometry. If the state admits a whole symmetry group, the gauges form a projective representation of that group on every bond.

Consider now a local (gauge) symmetry in a normal TN. If the symmetry is strictly one-local, it leaves each tensor invariant. As an illustration, for MPS, if
\begin{equation*}
  \begin{tikzpicture}[baseline=-0.1cm]
    \draw (0.5,0) rectangle (4.5,-0.5);
    \foreach \x/\t in {1/1,2/2,4/n}{
      \node[tensor,label=below:$A$] (t\x) at (\x,0) {};
      \draw[very thick] (t\x) --++ (0,0.5);
    }
    \node[tensor,label=left:$U$] (u) at (1,0.5) {};
    \draw[very thick] (u) --++ (0,0.5);    
    \node[fill=white] at (3,0) {$\dots$};
  \end{tikzpicture} = 
  \begin{tikzpicture}[baseline=-0.1cm]
    \draw (0.5,0) rectangle (4.5,-0.5);
    \foreach \x/\t in {1/1,2/2,4/n}{
      \node[tensor,label=below:$A$] (t\x) at (\x,0) {};
      \draw[very thick] (t\x) --++ (0,0.5);
    }
    \node[fill=white] at (3,0) {$\dots$};
  \end{tikzpicture} \ ,
\end{equation*}  
then 
\begin{equation*}
  \begin{tikzpicture}[baseline=-0.1cm]
    \draw (-0.5,0)--(0.5,0);
    \node[tensor,label=below:$A$] (t) at (0,0) {};
    \draw[very thick] (t)--(0,1);
    \node[tensor,label=left:$U$] (t) at (0,0.5) {};
  \end{tikzpicture}  =  
  \begin{tikzpicture}
    \draw (-0.5,0)--(0.5,0);
    \node[tensor,label=below:$A$\vphantom{$Z^{-1}$}] (t) at (0,0) {};
    \draw[very thick] (t)--(0,0.5);
  \end{tikzpicture} \ ,
\end{equation*}    
as looking at the rest of the tensors, we conclude that all gauges are the identity. For two-local symmetries, if  
\begin{equation*}
  \begin{tikzpicture}[baseline=-0.1cm]
    \draw (0.5,0) rectangle (4.5,-0.5);
    \foreach \x/\t in {1/1,2/2,4/n}{
      \node[tensor,label=below:$A$] (t\x) at (\x,0) {};
      \draw[very thick] (t\x) --++ (0,0.5);
    }
    \node[tensor,label=left:$U_L$] (ul) at (1,0.5) {};
    \draw[very thick] (ul) --++ (0,0.5);    
    \node[tensor,label=left:$U_R$] (ur) at (2,0.5) {};
    \draw[very thick] (ur) --++ (0,0.5);    
    \node[fill=white] at (3,0) {$\dots$};
  \end{tikzpicture} = 
  \begin{tikzpicture}[baseline=-0.1cm]
    \draw (0.5,0) rectangle (4.5,-0.5);
    \foreach \x/\t in {1/1,2/2,4/n}{
      \node[tensor,label=below:$A$] (t\x) at (\x,0) {};
      \draw[very thick] (t\x) --++ (0,0.5);
    }
    \node[fill=white] at (3,0) {$\dots$};
  \end{tikzpicture} \ ,
\end{equation*}  
then 
\begin{equation*}
  \begin{tikzpicture}[baseline=-0.1cm]
    \draw (-0.5,0)--(0.5,0);
    \node[tensor,label=below:$A$] (t) at (0,0) {};
    \draw[very thick] (t)--(0,1);
    \node[tensor,label=left:$U_L$] (t) at (0,0.5) {};
  \end{tikzpicture}  =  
  \begin{tikzpicture}[baseline=-0.1cm]
    \draw (-0.5,0)--(1,0);
    \node[tensor,label=below:$A$\vphantom{$Z^{-1}$}] (t) at (0,0) {};
    \node[tensor,label=below:$Z$\vphantom{$Z^{-1}$}] at (0.5,0) {};
    \draw[very thick] (t)--(0,0.5);
  \end{tikzpicture} \quad \text{and} \quad
  \begin{tikzpicture}[baseline=-0.1cm]
    \draw (-0.5,0)--(0.5,0);
    \node[tensor,label=below:$A$] (t) at (0,0) {};
    \draw[very thick] (t)--(0,1);
    \node[tensor,label=left:$U_R$] (t) at (0,0.5) {};
  \end{tikzpicture}  =  
  \begin{tikzpicture}[baseline=-0.1cm]
    \draw (-1,0)--(0.5,0);
    \node[tensor,label=below:$A$\vphantom{$Z^{-1}$}] (t) at (0,0) {};
    \node[tensor,label=below:$Z^{-1}$\vphantom{$Z^{-1}$}] at (-0.5,0) {};
    \draw[very thick] (t)--(0,0.5);
  \end{tikzpicture} \  .
\end{equation*}    
Here, if the state is symmetric under a whole group of unitaries, then the gauge $Z$ forms a linear representation of that group. Similar statements can be obtained for three-local symmetries as well as for any geometry provided that the TN satisfies the conditions in \cref{thm:normal}.

Consider now translation symmetry. We prove that a TI state (defined on a regular lattice) that has a normal PEPS description also has a TI PEPS description with the same bond dimension. This holds, for instance, for injective and normal 2D PEPS and MPS. Below we provide the proof for injective MPS, but the proof can easily be extended to the other cases as well. 

\begin{corollary}
  Let the tensors $A_i$ define an injective MPS such that the resulting state is translational invariant. Then all bond dimensions are the same and the state has a TI MPS description with an injective tensor $B$ that has the same bond dimension.
\end{corollary}

\begin{proof}
  Translational invariance means
  \begin{equation*}
    \ket{\Psi} = 
    \begin{tikzpicture}
      \draw (0.5,0) rectangle (4.5,-0.5);
      \foreach \x/\t in {1/1,2/2,4/n}{
        \node[tensor,label=below:$A_\t$] (t\x) at (\x,0) {};
        \draw[very thick] (t\x) --++ (0,0.5);
      }
      \node[fill=white] at (3,0) {$\dots$};
    \end{tikzpicture} = 
    \begin{tikzpicture}
      \draw (0.5,0) rectangle (4.5,-0.5);
      \foreach \x/\t in {1/2,2/3,4/1}{
        \node[tensor,label=below:$A_\t$] (t\x) at (\x,0) {};
        \draw[very thick] (t\x) --++ (0,0.5);
      }
      \node[fill=white] at (3,0) {$\dots$};
    \end{tikzpicture} \ ,
  \end{equation*}  
  thus, by \cref{thm:inj_MPS}, there are invertible matrices $Z_i$ such that for all $i$ ($n+1\equiv 1$)
  \begin{equation*}
    \begin{tikzpicture}
      \draw (-0.5,0)--(0.5,0);
      \node[tensor,label=below:$A_{i+1}$] (t) at (0,0) {};
      \draw[very thick] (t)--(0,0.5);
    \end{tikzpicture}  = 
    \begin{tikzpicture}
      \draw (-1,0)--(1,0);
      \node[tensor,label=below:$Z_i^{-1}$] at (-0.5,0) {};
      \node[tensor,label=below:\ $Z_{i+1}$\vphantom{$Z_i^{-1}$}] at (0.5,0) {};
      \node[tensor,label=below:$A_i$\vphantom{$Z_i^{-1}$}] (t) at (0,0) {};
      \draw[very thick] (t)--(0,0.5);
    \end{tikzpicture} \ .
  \end{equation*}  
  Therefore all tensors can be expressed with the help of the first tensor ($A_1$) together with some invertible matrices acting on the virtual dimensions of the tensor:
  \begin{equation}\label{eq:translation A_i}
    \begin{tikzpicture}
      \draw (-0.5,0)--(0.5,0);
      \node[tensor,label=below:$A_{i}$] (t) at (0,0) {};
      \draw[very thick] (t)--(0,0.5);
    \end{tikzpicture}  = 
    \begin{tikzpicture}
      \draw (-1,0)--(1,0);
      \node[tensor,label=below:$L_i^{-1}$] at (-0.5,0) {};
      \node[tensor,label=below:$R_i$\vphantom{$L_i^{-1}$}] at (0.5,0) {};
      \node[tensor,label=below:$A_1$\vphantom{$L_i^{-1}$}] (t) at (0,0) {};
      \draw[very thick] (t)--(0,0.5);
    \end{tikzpicture} \ ,
  \end{equation}  
  with 
  \begin{align*}
    L_i &= Z_1 Z_2 \dots Z_{i-1}, \\
    R_i &= Z_2 Z_3 \dots Z_i. 
  \end{align*}
  As $R_i L_{i+1}^{-1}= Z_1^{-1}$ for all $i$, substituting $A_i$ as in \cref{eq:translation A_i}, the state can be written as
  \begin{equation*}
    \ket{\Psi} =
    \begin{tikzpicture}
      \draw (0.5,0) rectangle (5.0,-0.5);
      \foreach \x/\t in {1/1,2/2,4/n}{
        \node[tensor,label=below:$A_1$] (t\x) at (\x,0) {};
        \node[tensor,label=above:$Z_1^{-1}$]  at (\x+0.5,0) {};
        \draw[very thick] (t\x) --++ (0,0.5);
      }
      \node[fill=white] at (3,0) {$\dots$};
    \end{tikzpicture} \ ,
  \end{equation*}
  where we have used that $A_{n+1}\equiv A_1$ and thus $R_{n+1} = L_{n+1} = 1$, which means that $R_n = Z_2 \dots Z_n = Z_1^{-1}$. This means that the state admits a TI MPS description with the tensor
  \begin{equation*}
    \begin{tikzpicture}
      \draw (-0.5,0)--(0.5,0);
      \node[tensor,label=below:$B$] (t) at (0,0) {};
      \draw[very thick] (t)--(0,0.5);
    \end{tikzpicture}  = 
    \begin{tikzpicture}
      \draw (-0.5,0)--(1,0);
      \node[tensor,label=below:\ $Z_1^{-1}$] at (0.5,0) {};
      \node[tensor,label=below:$A_1$\vphantom{$Z_1^{-1}$}] (t) at (0,0) {};
      \draw[very thick] (t)--(0,0.5);
    \end{tikzpicture} \ .
  \end{equation*}    
\end{proof}
The corresponding statement for injective 2D PEPS is
\begin{corollary}
  Let the tensors $A_{(i,j)}$ define an injective 2D PEPS such that the resulting state is translational invariant. Then all vertical (resp. all horizontal) bond dimensions are the same and the state has a TI 2D PEPS description with an injective tensor $B$ that has the same bond dimension.
\end{corollary}

\section{Conclusion}

In this paper we have shown the 'Fundamental Theorem' for injective and normal PEPS: two such tensor networks generate the same state if and only if the defining tensors are related with a local gauge. Moreover, the gauges relating the two set of tensors are uniquely defined up to a multiplicative constant. This result holds for a fixed (but large enough) system size. It is valid for any geometry, TI and non-TI setting, including 1D (MPS), 2D PEPS, higher dimensional PEPS, and other lattice geometries such as the honeycomb lattice, the Kagome lattice and the hyperbolic lattice used in the AdS/CFT correspondence. However, it does not include some important classes of TN like MERA, since they do not meet the conditions of \cref{thm:normal}.

\section{Acknowledgements}

We thank Barbara Kraus, David Sauerwein, Erez Zohar and Ilya Kull for inspiring conversations.

J.G.R. and  D.P.G. acknowledge financial support from MINECO (grant MTM2014-54240-P), Comunidad de Madrid (grant QUITEMAD+CM, ref. S2013/ICE-2801), and Severo Ochoa project SEV-2015-556. 

This project has received funding from the European Research Council (ERC) under
the European Union's Horizon 2020 research and innovation programme
through the ERC Starting Grant WASCOSYS (No.\ 636201), the ERC Consolidator Grant GAPS (No.\ 648913), and the ERC Advanced Grant QENOCOBA (No.\ 742102).
\appendix

\section{Normal MPS: alternative proof} \label{sec:normal_alt}

In \cref{sec:normal} we have shown that two normal TNs generate the same state if and only if the generating tensors are related with a gauge transformation. In the proof, we have blocked tensors to injective tensors. This proof is not optimal in the system size. For example, consider an MPS on five sites 
\begin{equation*}
  \ket{\Psi} = \ 
  \begin{tikzpicture}
    \draw (0.5,0) rectangle (5.5,-0.5);
    \foreach \x in {1,2,3,4,5}{
      \node[tensor,label=below:$A_\x$] (t\x) at (\x,0) {};
      \draw[very thick] (t\x) --++ (0,0.5);
    }
  \end{tikzpicture}  \ ,
\end{equation*}  
where the blocking of any two consecutive tensors:
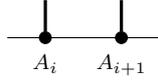
\begin{equation*}
  \begin{tikzpicture}
    \draw (0.5,0)--(2.5,0);
    \node[tensor,label=below:$A_{i}$] (a1) at (1,0) {};
    \draw[very thick] (a1)--++(0,0.5);
    \node[tensor,label=below:$A_{i+1}$] (a2) at (2,0) {};
    \draw[very thick] (a2)--++(0,0.5);
  \end{tikzpicture}
\end{equation*}
is injective. The proof in \cref{sec:normal} does not work for this case as this MPS cannot be blocked to a three-partite injective MPS (as it is too short). Here we prove a more size-efficient variant of \cref{lem:inj_isomorph} for this situation.

\Cref{lem:injective_union} implies that any region of at least size two is also injective. Now, similar to the injective case, for every edge and every matrix $X$ and $Y$, if 
\begin{equation*}
  \begin{tikzpicture}
    \draw (0.5,0) rectangle (5.5,-0.5);
    \foreach \x in {1,2,3,4,5}{
      \node[tensor,label=below:$A_\x$] (t\x) at (\x,0) {};
      \draw[very thick] (t\x) --++ (0,0.5);
    }
    \node[tensor, red,label=above:$X$] (x) at (2.5,0) {};
  \end{tikzpicture} = 
  \begin{tikzpicture}
    \draw (0.5,0) rectangle (5.5,-0.5);
    \foreach \x in {1,2,3,4,5}{
      \node[tensor,label=below:$A_\x$] (t\x) at (\x,0) {};
      \draw[very thick] (t\x) --++ (0,0.5);
    }
    \node[tensor, red,label=above:$Y$] (x) at (2.5,0) {};
  \end{tikzpicture} \ ,
\end{equation*}
then $X=Y$.

Consider now any virtual operation $X$ on a given edge:
\begin{equation*}
  \ket{\Psi'} = \begin{tikzpicture}
    \draw (0.5,0) rectangle (5.5,-0.5);
    \foreach \x in {1,2,3,4,5}{
      \node[tensor,label=below:$A_\x$] (t\x) at (\x,0) {};
      \draw[very thick] (t\x) --++ (0,0.5);
    }
    \node[tensor, red,label=above:$X$] (x) at (2.5,0) {};
  \end{tikzpicture} .
\end{equation*}
This operation can also be realized by three different two-local physical operators:
\begin{equation}\label{eq:normal_resonate}
  \ket{\Psi'} =  \begin{tikzpicture}
    \draw (0.5,0) rectangle (5.5,-0.5);
    \foreach \x in {1,2,3,4,5}{
      \node[tensor,label=below:$A_\x$] (t\x) at (\x,0) {};
      \draw[very thick] (t\x) --++ (0,0.5);
    }
    \node[longtensor, red,label=left:$O_1$] (o) at (1.5,0.5) {};
    \draw[very thick,red] (1,0.5)--++(0,0.5); 
    \draw[very thick,red] (2,0.5)--++(0,0.5); 
  \end{tikzpicture} = 
  \begin{tikzpicture}
    \draw (0.5,0) rectangle (5.5,-0.5);
    \foreach \x in {1,2,3,4,5}{
      \node[tensor,label=below:$A_\x$] (t\x) at (\x,0) {};
      \draw[very thick] (t\x) --++ (0,0.5);
    }
    \node[longtensor, red,label=left:$O_2$] (o) at (2.5,0.5) {};
    \draw[very thick,red] (2,0.5)--++(0,0.5); 
    \draw[very thick,red] (3,0.5)--++(0,0.5); 
  \end{tikzpicture} \  = \ 
  \begin{tikzpicture}
    \draw (0.5,0) rectangle (5.5,-0.5);
    \foreach \x in {1,2,3,4,5}{
      \node[tensor,label=below:$A_\x$] (t\x) at (\x,0) {};
      \draw[very thick] (t\x) --++ (0,0.5);
    }
    \node[longtensor, red,label=left:$O_3$] (o) at (3.5,0.5) {};
    \draw[very thick,red] (3,0.5)--++(0,0.5); 
    \draw[very thick,red] (4,0.5)--++(0,0.5); 
  \end{tikzpicture} \ , 
\end{equation}
with 
\begin{equation*}
   O_1 = 
  \begin{tikzpicture}
    \draw (-0.5,0) rectangle (1.5,1);
    \node[longtensor,label=below:$A_{23}^{-1}$] at (0.5,0) {};
    \node[tensor,label=above left:$A_{2}$] at (0,1) {};
    \node[tensor,label=above right:$A_{3}$] at (1,1) {};
    \node[tensor,label=right:$X$] at (1.5,0.5) {};
    \draw[very thick] (0,1)--++(0,0.5);
    \draw[very thick] (1,1)--++(0,0.5);
    \draw[very thick] (0,0)--++(0,-0.5);
    \draw[very thick] (1,0)--++(0,-0.5);
  \end{tikzpicture} \quad \text{and} \quad 
  O_2 = 
  \begin{tikzpicture}
    \draw (-0.5,0) rectangle (1.5,1);
    \node[longtensor,label=below:$A_{23}^{-1}$] at (0.5,0) {};
    \node[tensor,label=above left:$A_{2}$] at (0,1) {};
    \node[tensor,label=above right:$A_{3}$] at (1,1) {};
    \node[tensor,label=below:$X$] at (0.5,1) {};
    \draw[very thick] (0,1)--++(0,0.5);
    \draw[very thick] (1,1)--++(0,0.5);
    \draw[very thick] (0,0)--++(0,-0.5);
    \draw[very thick] (1,0)--++(0,-0.5);
  \end{tikzpicture} \ \quad \text{and} \quad
   O_3 = 
  \begin{tikzpicture}
    \draw (-0.5,0) rectangle (1.5,1);
    \node[longtensor,label=below:$A_{34}^{-1}$] at (0.5,0) {};
    \node[tensor,label=above left:$A_{3}$] at (0,1) {};
    \node[tensor,label=above right:$A_{4}$] at (1,1) {};
    \node[tensor,label=left:$X$] at (-0.5,0.5) {};
    \draw[very thick] (0,1)--++(0,0.5);
    \draw[very thick] (1,1)--++(0,0.5);
    \draw[very thick] (0,0)--++(0,-0.5);
    \draw[very thick] (1,0)--++(0,-0.5);
  \end{tikzpicture} \ .
\end{equation*}
Notice that both $X\mapsto O_1$ and $X\mapsto O_3^T$ are algebra homomorphisms, but the map $X\mapsto O_2$ not necessarily. Conversely: 
\begin{lemma}\label{lem:5}
Suppose that the state $\ket{\Psi'}$ can be written as
\begin{equation*}
  \ket{\Psi'} = \begin{tikzpicture}
    \draw (0.5,0) rectangle (5.5,-0.5);
    \foreach \x in {1,2,3,4,5}{
      \node[tensor,label=below:$A_\x$] (t\x) at (\x,0) {};
      \draw[very thick] (t\x) --++ (0,0.5);
    }
    \node[longtensor, red,label=left:$O_1$] (o) at (1.5,0.5) {};
    \draw[very thick,red] (1,0.5)--++(0,0.5); 
    \draw[very thick,red] (2,0.5)--++(0,0.5); 
  \end{tikzpicture} = 
  \begin{tikzpicture}
    \draw (0.5,0) rectangle (5.5,-0.5);
    \foreach \x in {1,2,3,4,5}{
      \node[tensor,label=below:$A_\x$] (t\x) at (\x,0) {};
      \draw[very thick] (t\x) --++ (0,0.5);
    }
    \node[longtensor, red,label=left:$O_2$] (o) at (2.5,0.5) {};
    \draw[very thick,red] (2,0.5)--++(0,0.5); 
    \draw[very thick,red] (3,0.5)--++(0,0.5); 
  \end{tikzpicture} \  = \ 
  \begin{tikzpicture}
    \draw (0.5,0) rectangle (5.5,-0.5);
    \foreach \x in {1,2,3,4,5}{
      \node[tensor,label=below:$A_\x$] (t\x) at (\x,0) {};
      \draw[very thick] (t\x) --++ (0,0.5);
    }
    \node[longtensor, red,label=left:$O_3$] (o) at (3.5,0.5) {};
    \draw[very thick,red] (3,0.5)--++(0,0.5); 
    \draw[very thick,red] (4,0.5)--++(0,0.5); 
  \end{tikzpicture} \ . 
\end{equation*}
Then there is a virtual operation $X$ on the bond $(2,3)$ such that
\begin{equation*}
  \ket{\Psi'} = \begin{tikzpicture}
    \draw (0.5,0) rectangle (5.5,-0.5);
    \foreach \x in {1,2,3,4,5}{
      \node[tensor,label=below:$A_\x$] (t\x) at (\x,0) {};
      \draw[very thick] (t\x) --++ (0,0.5);
    }
    \node[tensor, red,label=above:$X$] (x) at (2.5,0) {};
  \end{tikzpicture} \ ;
\end{equation*}
moreover, 
\begin{equation*}
  \begin{tikzpicture}
    \draw (0.5,0)--(2.5,0);
    \foreach \x in {1,2}{
      \node[tensor,label=below:$A_\x$] at (\x,0) {};
      \draw[very thick] (\x,0)--++(0,0.5);
    }
    \node[longtensor, red,label=left:$O_1$] (o) at (1.5,0.5) {};
    \draw[very thick,red] (1,0.5)--++(0,0.5); 
    \draw[very thick,red] (2,0.5)--++(0,0.5); 
  \end{tikzpicture} = 
  \begin{tikzpicture}
    \draw (0.5,0)--(3.0,0);
    \foreach \x in {1,2}{
      \node[tensor,label=below:$A_\x$] at (\x,0) {};
      \draw[very thick] (\x,0)--++(0,0.5);
    }
    \node[tensor, label=below:$X$] (x) at (2.5,0) {};
  \end{tikzpicture} \quad \text{and} \quad 
  \begin{tikzpicture}
    \draw (2.5,0)--(4.5,0);
    \foreach \x in {3,4}{
      \node[tensor,label=below:$A_\x$] at (\x,0) {};
      \draw[very thick] (\x,0)--++(0,0.5);
    }
    \node[longtensor, red,label=left:$O_3$] (o) at (3.5,0.5) {};
    \draw[very thick,red] (3,0.5)--++(0,0.5); 
    \draw[very thick,red] (4,0.5)--++(0,0.5); 
  \end{tikzpicture} \ = \     
  \begin{tikzpicture}
    \draw (2.0,0)--(4.5,0);
    \foreach \x in {3,4}{
      \node[tensor,label=below:$A_\x$] at (\x,0) {};
      \draw[very thick] (\x,0)--++(0,0.5);
    }
    \node[tensor, label=below:$X$] (y) at (2.5,0) {};
  \end{tikzpicture} \ ,    
\end{equation*}
and the maps $O_1\mapsto X$ and $O_3^T\mapsto X$ are uniquely defined and are algebra-homomorphisms. 
\end{lemma}

\begin{proof}
  Invert the injective tensor on the region $45$. We get
  \begin{equation}\label{eq:normal_act_123}
    \begin{tikzpicture}
      \draw (0.5,0)--(3.5,0);
      \foreach \x in {1,2,3}{
        \node[tensor,label=below:$A_\x$] at (\x,0) {};
        \draw[very thick] (\x,0)--++(0,0.5);
      }
      \node[longtensor, red,label=left:$O_1$] (o) at (1.5,0.5) {};
      \draw[very thick,red] (1,0.5)--++(0,0.5); 
      \draw[very thick,red] (2,0.5)--++(0,0.5); 
    \end{tikzpicture} = 
    \begin{tikzpicture}
      \draw (0.5,0)--(3.5,0);
      \foreach \x in {1,2,3}{
        \node[tensor,label=below:$A_\x$] at (\x,0) {};
        \draw[very thick] (\x,0)--++(0,0.5);
      }
      \node[longtensor, red,label=left:$O_2$] (o) at (2.5,0.5) {};
      \draw[very thick,red] (2,0.5)--++(0,0.5); 
      \draw[very thick,red] (3,0.5)--++(0,0.5); 
    \end{tikzpicture} \ .    
  \end{equation} 
  Similarly, inverting the tensor on the region $51$, we get
  \begin{equation}\label{eq:normal_act_234}
    \begin{tikzpicture}
      \draw (1.5,0)--(4.5,0);
      \foreach \x in {2,3,4}{
        \node[tensor,label=below:$A_\x$] at (\x,0) {};
        \draw[very thick] (\x,0)--++(0,0.5);
      }
      \node[longtensor, red,label=left:$O_2$] (o) at (2.5,0.5) {};
      \draw[very thick,red] (2,0.5)--++(0,0.5); 
      \draw[very thick,red] (3,0.5)--++(0,0.5); 
    \end{tikzpicture} = 
    \begin{tikzpicture}
      \draw (1.5,0)--(4.5,0);
      \foreach \x in {2,3,4}{
        \node[tensor,label=below:$A_\x$] at (\x,0) {};
        \draw[very thick] (\x,0)--++(0,0.5);
      }
      \node[longtensor, red,label=left:$O_3$] (o) at (3.5,0.5) {};
      \draw[very thick,red] (3,0.5)--++(0,0.5); 
      \draw[very thick,red] (4,0.5)--++(0,0.5); 
    \end{tikzpicture} \ .    
  \end{equation} 
  Therefore, plugging $A_4$ on the right side in \cref{eq:normal_act_123} and $A_1$ on the left side in \cref{eq:normal_act_234}, we get
  \begin{equation*}
    \begin{tikzpicture}
      \draw (0.5,0)--(4.5,0);
      \foreach \x in {1,2,3,4}{
        \node[tensor,label=below:$A_\x$] at (\x,0) {};
        \draw[very thick] (\x,0)--++(0,0.5);
      }
      \node[longtensor, red,label=left:$O_1$] (o) at (1.5,0.5) {};
      \draw[very thick,red] (1,0.5)--++(0,0.5); 
      \draw[very thick,red] (2,0.5)--++(0,0.5); 
    \end{tikzpicture} = 
    \begin{tikzpicture}
      \draw (0.5,0)--(4.5,0);
      \foreach \x in {1,2,3,4}{
        \node[tensor,label=below:$A_\x$] at (\x,0) {};
        \draw[very thick] (\x,0)--++(0,0.5);
      }
      \node[longtensor, red,label=left:$O_2$] (o) at (2.5,0.5) {};
      \draw[very thick,red] (2,0.5)--++(0,0.5); 
      \draw[very thick,red] (3,0.5)--++(0,0.5); 
    \end{tikzpicture} = 
    \begin{tikzpicture}
      \draw (0.5,0)--(4.5,0);
      \foreach \x in {1,2,3,4}{
        \node[tensor,label=below:$A_\x$] at (\x,0) {};
        \draw[very thick] (\x,0)--++(0,0.5);
      }
      \node[longtensor, red,label=left:$O_3$] (o) at (3.5,0.5) {};
      \draw[very thick,red] (3,0.5)--++(0,0.5); 
      \draw[very thick,red] (4,0.5)--++(0,0.5); 
    \end{tikzpicture} \ .    
  \end{equation*}   
  Comparing the two ends of the equation, similar to \cref{eq:inj_O->X_argument}, we get that
  \begin{equation*}
    \begin{tikzpicture}
      \draw (0.5,0)--(2.5,0);
      \foreach \x in {1,2}{
        \node[tensor,label=below:$A_\x$] at (\x,0) {};
        \draw[very thick] (\x,0)--++(0,0.5);
      }
      \node[longtensor, red,label=left:$O_1$] (o) at (1.5,0.5) {};
      \draw[very thick,red] (1,0.5)--++(0,0.5); 
      \draw[very thick,red] (2,0.5)--++(0,0.5); 
    \end{tikzpicture} = 
    \begin{tikzpicture}
      \draw (0.5,0)--(3.0,0);
      \foreach \x in {1,2}{
        \node[tensor,label=below:$A_\x$] at (\x,0) {};
        \draw[very thick] (\x,0)--++(0,0.5);
      }
      \node[tensor, label=below:$X$] (x) at (2.5,0) {};
    \end{tikzpicture} \quad \text{and} \quad 
    \begin{tikzpicture}
      \draw (2.5,0)--(4.5,0);
      \foreach \x in {3,4}{
        \node[tensor,label=below:$A_\x$] at (\x,0) {};
        \draw[very thick] (\x,0)--++(0,0.5);
      }
      \node[longtensor, red,label=left:$O_3$] (o) at (3.5,0.5) {};
      \draw[very thick,red] (3,0.5)--++(0,0.5); 
      \draw[very thick,red] (4,0.5)--++(0,0.5); 
    \end{tikzpicture} \ = \     
    \begin{tikzpicture}
      \draw (2.0,0)--(4.5,0);
      \foreach \x in {3,4}{
        \node[tensor,label=below:$A_\x$] at (\x,0) {};
        \draw[very thick] (\x,0)--++(0,0.5);
      }
      \node[tensor, label=below:$Y$] (y) at (2.5,0) {};
    \end{tikzpicture} \ .    
  \end{equation*}
Finally $X=Y$ by comparing the states they generate. These relations define $X$ uniquely and by composition, the maps $O_1\mapsto X$ and $O_3\mapsto X^T$ are algebra homomorphisms.
\end{proof}

Notice that similar to the injective case, this leads to
\begin{corollary}
  Let $A$ and $B$ be two normal TI MPS on $n\geq 2L+1$ sites with the property that blocking $L$ consecutive sites results in an injective tensor. Suppose they generate the same state:
  \begin{equation*}
    \ket{\Psi} = 
    \begin{tikzpicture}
      \draw (0.5,0) rectangle (4.5,-0.5);
      \foreach \x/\t in {1/1,2/2,4/n}{
        \node[tensor,label=below:$A$] (t\x) at (\x,0) {};
        \draw[very thick] (t\x) --++ (0,0.5);
      }
      \node[fill=white] at (3,0) {$\dots$};
    \end{tikzpicture} = 
    \begin{tikzpicture}
      \draw (0.5,0) rectangle (4.5,-0.5);
      \foreach \x/\t in {1/1,2/2,4/n}{
        \node[tensor,label=below:$B$] (t\x) at (\x,0) {};
        \draw[very thick] (t\x) --++ (0,0.5);
      }
      \node[fill=white] at (3,0) {$\dots$};
    \end{tikzpicture} \ .
  \end{equation*}  
  Then there is an invertible matrix $Z$ and a constant $\lambda$ with $\lambda^n=1$ such that 
    \begin{equation*}
      \begin{tikzpicture}
        \draw (-0.5,0)--(0.5,0);
        \node[tensor,label=below:$B$] (t) at (0,0) {};
        \draw[very thick] (t)--(0,0.5);
      \end{tikzpicture}  = \lambda \cdot
      \begin{tikzpicture}
        \draw (-1,0)--(1,0);
        \node[tensor,label=below:$Z^{-1}$] at (-0.5,0) {};
        \node[tensor,label=below:$Z$\vphantom{$Z^{-1}$}] at (0.5,0) {};
        \node[tensor,label=below:$A$\vphantom{$Z^{-1}$}] (t) at (0,0) {};
        \draw[very thick] (t)--(0,0.5);
      \end{tikzpicture} \ .
    \end{equation*}    
    Moreover the gauge $Z$ is unique up to a multiplicative constant.
\end{corollary}

The arguments in \cref{lem:5} can be applied for 2D PEPS as well. In the TI setting, this leads to
\begin{corollary}
  Let $A$ and $B$ be two normal 2D PEPS tensors such that every $L \times K$ region is injective. Suppose they generate the same state on some region $n\times m$ with $n \geq 2 L+1$ and $m\geq 2K+1$. Then $A$ and $B$ are related to each other with a gauge: 
  \begin{equation*}
    \begin{tikzpicture}[baseline=-0.1cm]
      \draw (-0.5,0)--(0.5,0);
      \draw (0,0,-0.5)--(0,0,0.5);
      \node[tensor,label=below:$B$] at (0,0) {};
      \draw[very thick] (0,0)--++(0,0.5);
    \end{tikzpicture} = \lambda \cdot 
    \begin{tikzpicture}[baseline=-0.1cm]
      \draw (-1.3,0)--(1.3,0);
      \draw (0,0,-1.6)--(0,0,1.6);
      \node[tensor,label=below:$A$] at (0,0) {};
      \node[tensor,label=right:$Y^{-1}$] at (0,0,-1) {};
      \node[tensor,label=below:$Y$] at (0,0,1) {};
      \node[tensor,label=below:$X^{-1}$] at (0.8,0,0) {};
      \node[tensor,label=below:$X$] at (-0.8,0,0) {};
      \draw[very thick] (0,0)--++(0,0.5);
    \end{tikzpicture} \ , 
  \end{equation*}
  with $\lambda^{n\cdot m} = 1$ and $X,Y$ invertible matrices. Moreover these matrices $X,Y$ are unique up to a multiplicative constant.
\end{corollary}

\begin{proof}[Sketch of proof]
  We only need to prove a statement similar to \cref{lem:5}. For that, notice that a virtual operation on a given bond can be interpreted as a physical operation on any of the following four regions (in the case of $K=L=2$):
    
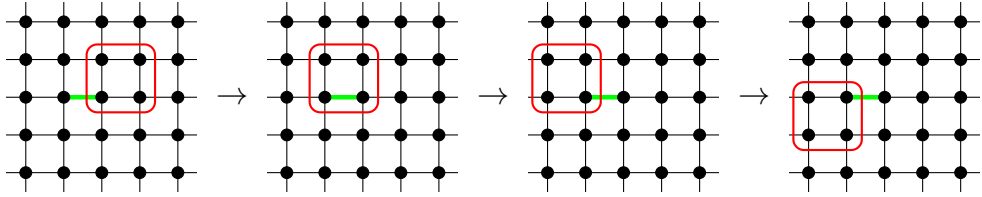
\begin{equation*}
      \begin{tikzpicture}
        \clip (0.25,0.25) rectangle (2.75,2.75);
        \draw[step=0.5] (-1,-1) grid (4,4);
        \draw[line width=0.6mm,green] (1,1.5)--(1.5,1.5);
        \foreach \x in {0,0.5,...,3}{
          \foreach \y in {0,0.5,...,3}{
            \node[tensor] at (\x,\y) {};
          }
        }
        \draw[red,thick, rounded corners] (1.3,1.3) rectangle (2.2,2.2);
      \end{tikzpicture} \ \rightarrow \ 
      \begin{tikzpicture}
        \clip (0.25,0.25) rectangle (2.75,2.75);
        \draw[step=0.5] (-1,-1) grid (4,4);
        \draw[line width=0.6mm,green] (1,1.5)--(1.5,1.5);
        \foreach \x in {0,0.5,...,3}{
          \foreach \y in {0,0.5,...,3}{
            \node[tensor] at (\x,\y) {};
          }
        }
        \draw[red,thick, rounded corners] (0.8,1.3) rectangle (1.7,2.2);
      \end{tikzpicture} \ \rightarrow \ 
      \begin{tikzpicture}
        \clip (0.25,0.25) rectangle (2.75,2.75);
        \draw[step=0.5] (-1,-1) grid (4,4);
        \draw[line width=0.6mm,green] (1,1.5)--(1.5,1.5);
        \foreach \x in {0,0.5,...,3}{
          \foreach \y in {0,0.5,...,3}{
            \node[tensor] at (\x,\y) {};
          }
        }
        \draw[red,thick, rounded corners] (0.3,1.3) rectangle (1.2,2.2);
      \end{tikzpicture} \ \rightarrow \ 
      \begin{tikzpicture}
        \clip (0.25,0.25) rectangle (2.75,2.75);
        \draw[step=0.5] (-1,-1) grid (4,4);
        \draw[line width=0.6mm,green] (1,1.5)--(1.5,1.5);
        \foreach \x in {0,0.5,...,3}{
          \foreach \y in {0,0.5,...,3}{
            \node[tensor] at (\x,\y) {};
          }
        }
        \draw[red,thick, rounded corners] (0.3,0.8) rectangle (1.2,1.7);
      \end{tikzpicture} \ 
    \end{equation*} 
  We need to prove that conversely, any four physical operators on the above regions that transforms the PEPS into the same state means that the transformation can equally be done with a virtual operation on the highlighted edge. The system size required to compare any two consecutive regions is only $5\times 5$ (and in general, $(2L+1)\times (2K+1)$). Therefore, similar to \cref{lem:5},
    \begin{equation*}
      \begin{tikzpicture}
        \clip (0.25,0.75) rectangle (2.25,2.25);
        \draw[step=0.5] (-1,-1) grid (4,4);
        \draw[line width=0.6mm,green] (1,1.5)--(1.5,1.5);
        \foreach \x in {0,0.5,...,3}{
          \foreach \y in {0,0.5,...,3}{
            \node[tensor] at (\x,\y) {};
          }
        }
        \draw[red,thick, rounded corners] (1.3,1.3) rectangle (2.2,2.2);
      \end{tikzpicture} \ = \ 
      \begin{tikzpicture}
        \clip (0.25,0.75) rectangle (2.25,2.25);
        \draw[step=0.5] (-1,-1) grid (4,4);
        \draw[line width=0.6mm,green] (1,1.5)--(1.5,1.5);
        \foreach \x in {0,0.5,...,3}{
          \foreach \y in {0,0.5,...,3}{
            \node[tensor] at (\x,\y) {};
          }
        }
        \draw[red,thick, rounded corners] (0.8,1.3) rectangle (1.7,2.2);
      \end{tikzpicture} \ = \  
      \begin{tikzpicture}
        \clip (0.25,0.75) rectangle (2.25,2.25);
        \draw[step=0.5] (-1,-1) grid (4,4);
        \draw[line width=0.6mm,green] (1,1.5)--(1.5,1.5);
        \foreach \x in {0,0.5,...,3}{
          \foreach \y in {0,0.5,...,3}{
            \node[tensor] at (\x,\y) {};
          }
        }
        \draw[red,thick, rounded corners] (0.3,1.3) rectangle (1.2,2.2);
      \end{tikzpicture} \ = \ 
      \begin{tikzpicture}
        \clip (0.25,0.75) rectangle (2.25,2.25);
        \draw[step=0.5] (-1,-1) grid (4,4);
        \draw[line width=0.6mm,green] (1,1.5)--(1.5,1.5);
        \foreach \x in {0,0.5,...,3}{
          \foreach \y in {0,0.5,...,3}{
            \node[tensor] at (\x,\y) {};
          }
        }
        \draw[red,thick, rounded corners] (0.3,0.8) rectangle (1.2,1.7);
      \end{tikzpicture} \ ,
    \end{equation*} 
    with open boundaries. Compare now the first and the last expression in the above equation. One can add two-two tensors in the upper left and lower right corner and invert the resulting regions, leading to
    \begin{equation*}
      \begin{tikzpicture}
        \clip (0.25,0.75) rectangle (2.25,2.25);
        \draw[step=0.5] (-1,-1) grid (4,4);
        \draw[line width=0.6mm,green] (1,1.5)--(1.5,1.5);
        \foreach \x in {0,0.5,...,3}{
          \foreach \y in {0,0.5,...,3}{
            \node[tensor] at (\x,\y) {};
          }
        }
        \draw[red,thick, rounded corners] (1.3,1.3) rectangle (2.2,2.2);
        \filldraw[white] (0.25,1.75) rectangle (1.25,2.25);
        \filldraw[white] (1.25,0.75) rectangle (2.25,1.25);
      \end{tikzpicture} \ = \ 
      \begin{tikzpicture}
        \clip (0.25,0.75) rectangle (2.25,2.25);
        \draw[step=0.5] (-1,-1) grid (4,4);
        \draw[line width=0.6mm,green] (1,1.5)--(1.5,1.5);
        \foreach \x in {0,0.5,...,3}{
          \foreach \y in {0,0.5,...,3}{
            \node[tensor] at (\x,\y) {};
          }
        }
        \draw[red,thick, rounded corners] (0.3,0.8) rectangle (1.2,1.7);
        \filldraw[white] (0.25,1.75) rectangle (1.25,2.25);
        \filldraw[white] (1.25,0.75) rectangle (2.25,1.25);
      \end{tikzpicture} \ .
    \end{equation*}
    This results in the desired virtual operation on the highlighted edge. The rest of the proof is the same as the proof of \cref{thm:3}.
\end{proof}

\bibliography{../../library}

\end{document}